\documentclass{article}
\usepackage{graphicx} 
\usepackage[utf8]{inputenc}
\usepackage{amsmath}
\usepackage{amssymb}
\usepackage{amsthm}
\usepackage{dsfont}
\usepackage[a-1b]{pdfx} 
\usepackage{hyperref}
\usepackage{csquotes}
\usepackage{pdflscape}
\usepackage[ruled]{algorithm2e}
\usepackage[a4paper, total={6in, 8.5in}]{geometry}
\hypersetup{colorlinks,linkcolor={black},citecolor={blue},urlcolor={black}}  
\usepackage[english]{babel}
\usepackage{comment}
\usepackage{bbm}
\usepackage{float}
\usepackage{multirow}
\usepackage{etoolbox}

\usepackage[round]{natbib}
\usepackage[toc,page]{appendix}

\usepackage{authblk}

\newtheorem{theorem}{Theorem}[section]
\newtheorem{corollary}[theorem]{Corollary}
\newtheorem{lemma}[theorem]{Lemma}
\newtheorem{proposition}[theorem]{Proposition}

\newtheoremstyle{mystyle}
  {\topsep}
  {\topsep}
  {\normalfont}
  {}
  {\bfseries}
  {.}
  {.5em}
  {}

\theoremstyle{mystyle}
\newtheorem{remark}[theorem]{Remark}

\newcommand{\pr}{\mathbb{P}}

\newcommand{\Ev}{\mathbb{E}}
\newcommand{\ind}{\mathds{1}}

\newcommand{\real}{\mathbb{R}}
\newcommand{\loss}{\mathcal{L}}
\newcommand{\set}{\mathcal{S}}
\newcommand{\setc}{\mathcal{C}}
\newcommand{\mom}{\mathrm{MoM}}
\newcommand{\momom}{\mathrm{MoMoM}}
\newcommand{\binomial}{\textrm{Binom}}

\title{Merging uncertainty sets via majority vote}
\author[1]{Matteo Gasparin}
\author[2,3]{Aaditya Ramdas}
\affil[3]{Machine Learning Department, Carnegie Mellon University}
\affil[2]{Department of Statistics and Data Science, Carnegie Mellon University}
\affil[1]{Department of Statistical Sciences, University of Padova}
\date{}
\begin{document}

\date{\today}
\maketitle

\begin{abstract}
Given $K$ uncertainty sets that are arbitrarily dependent --- for example, confidence intervals for an unknown parameter obtained with $K$ different estimators, or prediction sets obtained via conformal prediction based on $K$ different algorithms on shared data --- we address the question of how to efficiently combine them in a black-box manner to produce a single uncertainty set. We present a simple and broadly applicable majority vote procedure that produces a merged set with nearly the same error guarantee as the input sets. We then extend this core idea in a few ways: we show that weighted averaging can be a powerful way to incorporate prior information, and a simple randomization trick produces strictly smaller merged sets without altering the coverage guarantee. Further improvements can be obtained if the sets are exchangeable. 
We also show that many modern methods, like split conformal prediction, median of means, HulC and cross-fitted ``double machine learning'', can be effectively derandomized using these ideas.
\end{abstract}

\tableofcontents

\section{Introduction}
\label{sec:intro}
Uncertainty quantification is a cornerstone of statistical science and is now rapidly gaining prominence within machine learning. 
For example, the development of conformal prediction~\citep{vovk2005} has been instrumental in recent years, which is a method to construct prediction sets with a coverage guarantee under weak distributional assumptions. 

This paper develops methods for combining $K$ different uncertainty sets (e.g.\ prediction or confidence sets) that are arbitrarily dependent (perhaps due to shared data) in order to obtain a single set with nearly the same coverage. 
As one motivation, consider $K$  ``agents'' that process some private and some public data in different ways in order to define $K$ uncertainty sets. Their use of the public data in unknown ways causes the sets to be dependent. Our work studies ways to combine these sets, despite unknown dependence. 

Formally, we start with a collection of $K$ different sets $\setc_k$ (one from each agent), each having a confidence level $1-\alpha$ for some $\alpha \in (0,1)$:
\begin{equation}
\label{eq:coverage}
\pr\big(c \in \setc_k\big) \geq 1-\alpha, \quad k=1, \dots, K,
\end{equation}
where $c$ denotes our target (e.g.\ an outcome that we want to predict, or some underlying functional of the data distribution). We say that $\setc_k$ has \emph{exact} coverage if $\pr(c \in \setc_k) = 1-\alpha$.
Since the sets $\setc_k$ are based on data, they are random quantities by definition, but $c$ can be either fixed or random; for example, in the case of confidence sets for a target functional of a distribution it is fixed, but it is random in the case of prediction sets for an outcome (e.g.\ conformal prediction). Our method will be agnostic to such details.

Our objective as the ``aggregator'' of uncertainty is to combine the sets in a black-box manner in order to create a new set that exhibits favorable properties in both coverage and size. A first (trivial) solution is to define the set $\setc^J$ as the union of the others:
\begin{equation*}
    \setc^J = \bigcup_{k=1}^K \setc_k.
\end{equation*}
Clearly, $\setc^J$ respects the property defined in \eqref{eq:coverage}, but the resulting set is typically too large and has significantly inflated coverage. On the other hand, the set resulting from the intersection $\setc^I = \bigcap_{k=1}^K \setc_k$ is narrower, but typically has inadequate coverage --- it guarantees at least $1-K\alpha$ coverage by the Bonferroni inequality \citep{bonferroni1936}, but this is uninformative when $K$ is large.

The current paper addresses the setting where only a single interval is known from each agent.\footnote{If the aggregator knows the $(1-\alpha)$-confidence intervals for \emph{every} $\alpha \in (0,1)$, they can combine these using well-known rules to combine dependent p-values~\citep{vovk2020} (see Appendix~\ref{sec:merg_conf_dist}).} 
A major motivation is aesthetic: while the literature on ``black-box'' combination of p-values (that is, without access to the underlying data) is rich --- dating from Fisher's famous combination rule to modern treatments~\citep{vovk2020} --- analogous results for uncertainty sets like confidence or prediction intervals are less studied or known.
A second motivation is practical --- while in some situations it is possible to obtain the p-values for every possible target value of $c$, there are many cases where this is intractable, and it is easier to simply obtain the confidence intervals at a fixed $1-\alpha$ level. Some examples of the latter type include procedures that are tuned to a particular level $\alpha$, like median-of-means~\citep{lugosi2019mean} or HulC \citep{kuchibhotla2021hulc}. In yet other cases, for example in differential privacy or distributed/federated learning, the aggregator may only have access to the $K$ different intervals and not the p-values or the raw data due to privacy constraints \citep{humbert2023, federatedLearning,waudbysmith2023}. 

In the following, we will define new aggregation schemes based on the simple concept of voting, which can be used to merge confidence or prediction sets. Section~\ref{sec:majority_vote} presents our general methodology for constructing the sets. In Section~\ref{sec:derand}, we explain how the majority vote method can be used in order to \emph{derandomize} statistical procedures based on data splitting.
In Section~\ref{sec:conformal} we apply our procedure in the context of conformal inference.

\section{Majority vote: definition, extensions and properties}\label{sec:majority_vote}

We now study a versatile method for combining uncertainty sets that is evidently broadly applicable. The key idea is based  on voting: each agent's interval gets to vote, and each point in the space of interest (where the target $c$ lies) will be part of the final set if it is vouched for by more than half (or more generally, some fraction) of the agents.

Majority voting is popular in other contexts, like ensemble methods for prediction; see \cite{breiman1996} and \cite{kuncheva2003,kuncheva2014}. The idea has been proposed within the context of combining conformal prediction intervals by \cite{cherubin2019} and \cite{solari2022} (though the latter work does not cite the earlier ones). 

This section compiles the relevant results in a succinct manner, and building on these, we extend the method in multiple directions. Specifically, we allow for the incorporation of a priori information, and additionally, we are able to achieve smaller sets through the use of a simple randomization or permutation technique without altering the coverage properties. 

The majority vote procedure for sets can be seen as dual to the results in \cite{ruger1978}, who presented a method for combining $K$ different p-values for testing a null hypothesis based on their order statistics; see Appendix~\ref{sec:merg_conf_dist}. These results are used, for example, in multi-split inference where the single agent wants to reduce the randomness induced by sample-splitting by performing many random splits and combining the results; see \cite{diciccio2020}. 

Recently, \cite{guo2023} introduced a different subsampling-based method to conduct inference in the case of multiple splits. Their results assume exchangeability of the underlying sets or p-values.
We handle both the exchangeable and non-exchangeable cases, but the more important distinction is that our method is ``black-box'' (needing to know no details of how the original sets were constructed) while theirs is not. Further, their method needs more distributional assumptions for their asymptotics to hold, while ours is nonasymptotically valid under no assumptions. Finally, their methods are computationally intensive and sophisticated to understand and use, while ours are simple in comparison. In turn, one may expect that when one has access to the full data to perform subsampling, and their distributional assumptions are true, their method may be more efficient than ours.

\subsection{The majority vote procedure}
\label{sec:gen_maj_vote}
Let the observed data $z=(z_1, \dots, z_n)$ be a realization of the random vector $Z=(Z_1, \dots, Z_n)$. In particular, $z=(z_1, \dots, z_n)$ is a point in a sample space $\mathcal{Z}$, while our target $c$ is a point in a measurable space $(\set, \mathcal{A}, \nu)$, where $\mathcal{A}$ is a $\sigma$-algebra on $\set$ and $\nu$ is a measure on $\set$. As mentioned earlier, it is important to note that $c$ can itself be a random variable. The sets $\setc_k = \setc_k(z) \subseteq \set$, $k=1, \dots, K$, based on the observed data, follow the property \eqref{eq:coverage}, where the probability refers to the joint distribution $(Z,c)$ (or only $Z$ if $c$ is fixed). Naturally, the different sets may have only been constructed using different subsets of $z$ (as different public and private data may be available to each of the agents). Let us define a new set $\setc^M$ that includes all points \emph{voted} by at least half of the sets:
\begin{equation}
\label{eq:cm}
    \setc^M:=\bigg\{s \in \set: \frac{1}{K} \sum_{k=1}^K \ind\{s \in \setc_k\} > \frac{1}{2}  \bigg\}.
\end{equation}
The following result stems from~\cite{kuncheva2003} and~\cite{cherubin2019}, and again later by \cite{solari2022}, but we provide a short proof to be self-contained. 
\begin{theorem}
\label{th:CM}
    Let $\setc_1, \dots, \setc_K$ be $K \geq 2$ different sets satisfying property \eqref{eq:coverage}. Then, the set $\setc^M$  defined in \eqref{eq:cm} has coverage of at least $1-2\alpha$:
    \begin{equation}
    \label{eq: coverage_cm}
    \pr\big(c \in \setc^M\big) \geq 1-2\alpha.
    \end{equation}
\end{theorem}
\begin{proof}
    Let $\phi_k=\phi_k(Z,c)=\ind\{c \notin \setc_k\}$ be a Bernoulli random variable such that $\Ev[\phi_k] \leq \alpha$, $k=1,\dots,K$.     We have by Markov's inequality,
    \[
    \begin{split}
    \pr(c \notin \setc^M) &= \pr\left(\frac{1}{K} \sum_{k=1}^K \phi_k \geq \frac{1}{2} \right) \leq 2 \Ev\left[\frac{1}{K} \sum_{k=1}^K \phi_k \right] = \frac{2}{K} \sum_{k=1}^K \Ev\left[\phi_k \right] \leq 2\alpha ,
    \end{split}
    \]
  which concludes the proof.
\end{proof}

The above proof essentially relies on defining e-values $E_k = \phi_k/\alpha$, combining them by averaging, and applying Markov's inequality~\citep{ramdas2024hypothesis}, but we omit further explicit reference to e-values to keep things simple for the unacquainted reader.

The merged set could be the empty set --- however, it cannot be empty with probability more than $2\alpha$, because it would otherwise violate the coverage guarantee. In practice, it occurs with much smaller probability, almost never occurring in our simulations.

\begin{remark}
    Actually, a slightly tighter bound can be obtained if $K$ is odd.
    In this case, for a point to be contained in $\setc^M$, it must be voted for by at least $\lceil K/2 \rceil$ of the other sets. This implies that, with the same arguments used in Theorem~\ref{th:CM}, the probability of miscoverage is equal to $\alpha K/\lceil K/2 \rceil = 2\alpha K/(K+1)$, approaching \eqref{eq: coverage_cm} for large $K$.
    
    This result is known to be tight in a worst-case sense; a simple example from \cite{kuncheva2003} shows that if $K$ is odd and if the sets have a particular joint distribution, then the error will equal $(\alpha K)/\lceil K/2 \rceil$. This worst-case distribution allows for only two types of cases: either all agents provide the same set that contains $c$ (so majority vote is correct), or $\lfloor K/2 \rfloor$ sets contain $c$ but the others do not (so majority vote is incorrect). Each of the latter cases happens with some probability $p$, so the probability that majority vote makes an error is $\binom{K}{\lfloor K/2 \rfloor +1} p$. The probability that any particular agent makes an error is $\binom{K-1}{\lfloor K/2 \rfloor}p$, which we set as our choice of $\alpha$, and then we see that the probability of error for majority vote simplifies to $\alpha K/\lceil K/2 \rceil$.  \emph{Despite the apparent tightness of majority vote in the worst-case, we will develop several ways to improve this procedure in non-worst-case instances, while retaining the same worst-case performance.}
\end{remark}
\begin{remark}[When does majority vote overcover and when does it undercover?]
    While the worst case theoretical guarantee for majority vote is a coverage level of $1-2\alpha$, sometimes it will get close to the desired $1-\alpha$ coverage, and sometimes it may even overcover, achieving coverage closer to one. Here, we provide some intuition for when to expect each type of behavior in practice assuming $\alpha < 1/2$, foreshadowing many results to come. If the sets are actually independent (or nearly so), we should expect the method to have coverage more than $1-\alpha$. This can be seen via an application of Hoeffding's inequality in place of Markov's inequality in the proof of Theorem~\ref{th:CM}: since each $\phi_k$ has expectation (at most) $\alpha$, we should expect $\frac{1}{K} \sum_{k=1}^K \phi_k$ to concentrate around $\alpha$, and the probability that this average exceeds $1/2$ is exponentially small (as opposed to $2\alpha$), being at most $\exp(-2K(1/2-\alpha)^2)$ by Hoeffding's inequality. In contrast, if the sets are identical (the opposite extreme of independence), clearly the method has coverage $1-\alpha$. As argued in the previous remark, there is a worst case dependence structure that forces majority to vote to have an error of (essentially) $2\alpha$. Finally, if the sets are exchangeable, it appears more likely that the coverage will be closer to $1-\alpha$ than $1-2\alpha$. While one informal reason may be that exchangeability connects the two extremes of independence and being identical (with coverages close to $1$ and $1-\alpha$), a slightly more formal  reason is that under exchangeability, we will later in this section actually devise a strictly tighter set $\setc^E$ than $\setc^M$ which also achieves the same coverage guarantee of $1-2\alpha$, thus making $\setc^M$ itself likely to have a substantially higher coverage.
\end{remark}

\subsection{Other thresholds and upper bounds}

The above method and result can be easily generalized beyond the threshold value of 1/2. We record it as a result for easier reference. For any $\tau \in [0,1)$, let
\begin{equation}\label{eq:c^tau}
    \setc^\tau:=\bigg\{s \in \set: \frac{1}{K} \sum_{k=1}^K \ind\{s \in \setc_k\} > \tau \bigg\}.
\end{equation}

\begin{theorem}\label{th:ctau}
     Let $\setc_1, \dots, \setc_K$ be $K \geq 2$ different sets satisfying property
     \eqref{eq:coverage}. Then, 
    \begin{equation*}
    \pr\big(c \in \setc^\tau \big) \geq 1-\alpha/(1-\tau).
    \end{equation*}
\end{theorem}
The proof follows the same lines as the original results outlined in Theorem~\ref{th:CM} and is thus omitted. As expected, it can be noted that the obtained bounds decrease as $\tau$ increases. In fact, for larger values of $\tau$, smaller sets will be obtained. One can check that this result also yields the right bound for the intersection ($\tau=1 - 1/K$) and the union ($\tau=0$). In certain situations, it is possible to identify an upper bound to the coverage of the set resulting from the majority vote.

\begin{theorem}
\label{th:up_bound}
    Let $\setc_1, \dots, \setc_K$ be $K \geq 2$ different sets having exact coverage $1-\alpha$. Then, 
    \begin{equation}
    \label{eq:upper_bound}
        \pr(c \in \setc^M) \leq 1 - \frac{K\alpha - \left \lceil{\frac{K}{2}}\right\rceil + 1}{K - \left \lceil{\frac{K}{2}}\right\rceil  + 1}.
    \end{equation}
\end{theorem}
A similar bound can be derived if the miscoverage of the input sets is not exact but is upper-bounded by $\alpha$. For typically employed values of $\alpha$, this upper bound is useful only for small $K$. When $K=2$, it can be seen that \eqref{eq:cm} coincides with the intersection between the two sets; this correctly implies that the coverage in this case lies in $[1-2\alpha, 1-\alpha]$.

\subsection{On the computation of the majority vote set}

Even if the input sets are intervals, the majority vote set may be a union of intervals. 
In Appendix \ref{sec:algorithm}, we describe a simple aggregation algorithm to find this set quickly by sorting the endpoints of the input intervals and checking some simple conditions. But in practice, we find that it is indeed a non-empty interval in the vast majority of our simulations, suggesting that it may be possible to identify some sufficient conditions under which this happens. One such condition is given below, and it is represented in Figure~\ref{fig:hist}.
\begin{lemma}\label{lemma:inter}
    If $\setc_1,\dots,\setc_K$ are one-dimensional intervals and $\cap_{k=1}^K \setc_k \neq \varnothing$, then $\setc^\tau$ is an interval for any $\tau$, and, in particular, $\setc^M$ is an interval.
\end{lemma}

A typical example of a non-empty intersection arises when all intervals contain the target $c$, and, as mentioned earlier, this probability is at least $1-K\alpha$. 
Clearly, this is an underestimate of the true probability and it is non-trivial only for moderate values of $K$. One can avoid the problem of having an union of intervals by taking the convex hull of set (the smallest interval containing the set) and this solution is used for example in \cite{gupta2022} in the framework of leave-one-out conformal prediction. 

\begin{figure}
    \centering
    \includegraphics[width=0.8\textwidth]{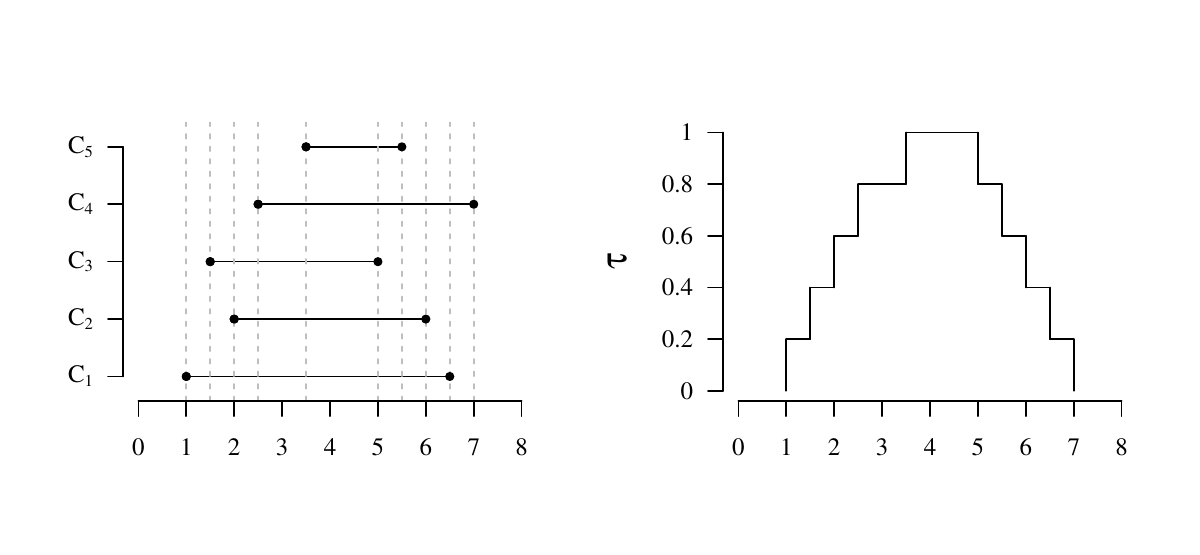}
    \caption{Visual representation of the majority vote procedure when $\cap_{k=1}^K \setc_k \neq \varnothing$.}
    \label{fig:hist}
\end{figure}

\subsection{Equal-width intervals and the ``Median of Midpoints''}

A special case for the majority vote set in \eqref{eq:cm} occurs if the input sets are (or can be bounded by) intervals of the same width, and can thus each can be represented as their midpoint plus/minus a half-width, then the following rule results in a more computationally efficient procedure. We call this method the ``Median of Midpoints''.

\begin{theorem}\label{thm:median-of-midpoints}
    Suppose the input sets $\setc_1,\dots,\setc_K$ are intervals having the same width, and let $\setc^M$ be their majority vote set. Let the midpoint of $\setc_k$ be denoted by $c_k$. 
    If $K$ is odd, let $c_{(\lceil K/2 \rceil)}$ denote their median, and let $\setc^{(K/2)}$ denote the interval whose midpoint is $c_{(\lceil K/2 \rceil)}$. 
    If $K$ is even, let $\setc^{(K/2)}$ be defined by the intersection of the sets whose midpoints are $c_{(K/2)}$ and $c_{(1+ K/2)}$.
    Then, $\setc^{(K/2)} \supseteq \setc^M$ and is hence also a $1-2\alpha$ uncertainty set. Further, $\setc^{(K/2)}$ has at most the same width as the input sets. If $\cap_{k=1}^K \setc_k \neq \varnothing$, then $\setc^{(K/2)}=\setc^M$.
 \end{theorem}

\subsection{How large is the majority vote set?}

One naive way to combine the $K$ sets is to randomly select one of them as the final set; this method clearly has coverage $1-\alpha$, and its length is in between their union and intersection, so it seems reasonable to ask how it compares to majority vote. Surprisingly, majority vote is not always strictly better than this approach in terms of the expected length of the set: consider, for example, three nested intervals $\setc_1, \setc_2, \setc_3$ of width $10, 8$ and $3$, respectively. The majority vote set is $\setc_2$, with a length of $8$, but randomly selecting an interval results in an average length of $7$. However, we show next that the majority vote set cannot be more than twice as large. In addition, when the input sets are intervals, it is never wider than the largest interval. 

\begin{theorem}
    \label{thm:length_mv}
    Let $\nu(\setc^\tau)$ be the measure (size) of $\setc^\tau$ in \eqref{eq:c^tau}. Then, for all $\tau \in [0,1)$,
    \begin{equation}\label{eq:m^tau}
        \nu(\setc^\tau) \leq \frac{1}{K\tau}\sum_{k=1}^K \nu(\setc_k).
    \end{equation}
    If the input sets are $K$ one-dimensional intervals, for all $\tau \in [\frac{1}{2}, 1)$, we have that
    \begin{equation}\label{eq:m^tau_max}
        \nu(\setc^\tau) \leq \max_k \nu(\setc_k).
    \end{equation}
\end{theorem}
In \eqref{eq:m^tau}, $\nu$ could be the Lebesgue measure (for intervals), or the counting measure (for discrete, categorical sets), for example.
The proof of \eqref{eq:m^tau_max} uses the fact that the majority vote set is elementwise monotonic in its input sets, meaning that if any of the input sets gets larger, the majority vote set can never get smaller. From \eqref{eq:m^tau} we have that if $\tau=1/2$, then the  measure of the majority vote set is never larger than 2 times the average of the measure of the initial sets. This result is essentially tight as can be seen in the following example involving one-dimensional intervals. For odd $K$, let $(K+1)/2$ intervals have a length $L$, while the rest have a length of nearly 0. The average length is then $(K+1)L/(2K)$, and the majority vote has length $L$, whose ratio approaches $1/2$ for large $K$. Clearly, this is just an illustrative example pointing out the tightness of the obtained bound; in practice, we usually observe a less adversarial behavior. In addition, \eqref{eq:m^tau} gives the right bound for the intersection ($\tau \uparrow 1$) and for the union ($\tau \uparrow 1/K$).

Also the bound in \eqref{eq:m^tau_max} is tight as can be seen from the following example. Consider a scenario similar to that of the last example. Suppose we have $K > 2$ confidence intervals, with odd $K$, such that $\setc_1= \dots = \setc_{\lceil K/2 \rceil} = (-\epsilon,1)$ and $\setc_{\lceil K/2 \rceil + 1} = \dots = \setc_K = (0,1+\epsilon)$, for some $\epsilon>0$. By definition, if $\tau=(K-2)/(2K)<1/2$, then the set $\setc^\tau$ coincides with the union of the initial sets, which is larger than the input intervals and has width $1+2\epsilon$. On the other hand, choosing $\tau=1/2$, as a consequence of Theorem~\ref{thm:median-of-midpoints} we obtain a set whose size is $1+\epsilon$, matching the width of the input sets. This fact provides a practical justification for choosing $\tau = 1/2$: it is the smallest $\tau$ for which the combined set cannot be larger than the input sets. In particular, a simple majority vote seems to offer a good compromise between coverage and size.

\subsection{Combining independent or nested confidence sets}\label{sec:independent_sets}

When $\setc_1, \dots, \setc_K$ are independent, then a more tailored combination rule can yield a smaller set with exactly $1-\alpha$ coverage. The rule is similar to~\eqref{eq:cm}, albeit with a different threshold, which is related to the quantile of a binomial distribution with $K$ trials and parameter $1-\alpha$. We define $Q_K(\alpha)$ as the $\alpha$-quantile of a $\binomial(K,1-\alpha)$:
\begin{equation}
\label{eq:quantile}
    Q_K(\alpha):=\sup\{x:F(x)\leq \alpha\},
\end{equation}
where $F(\cdot)$ is the cumulative distribution function of a $\binomial(K, 1-\alpha)$. 

\begin{proposition}
\label{prop:indep}
    Let $\setc_1, \dots, \setc_K$ be $K \geq 2$ different independent sets satisfying~\eqref{eq:coverage} and for a fixed parameter of interest $c$. Then, the set
    \begin{equation*}
        \setc^M=\left\{s \in \set: \sum_{k=1}^K \ind\{s \in \setc_K\} > Q_K(\alpha) \right\}
    \end{equation*}
    has coverage equal to $1-\alpha$.
\end{proposition}
We emphasize that this result requires that $c$ be a fixed quantity. If $c$ were to be random, the independence between the events $\ind\{c \in \setc_k\}$ and $\ind\{c \in \setc_l\}$, with $k \neq l$, would be compromised even if the sets were based on independent observations. 

Another (trivial) special case where it is possible to achieve a coverage of level $1-\alpha$ appears when the sets are almost surely nested (not necessarily independent). Let us suppose that $\setc_1 \subseteq \dots \subseteq \setc_K$ holds almost surely and we obtain the set $\setc^M$ as in~\eqref{eq:cm}. By definition, all the points contained in $\setc_1$ will be part of the set $\setc^M$, which implies that $\setc^M$ is a set with confidence level equal to $1-\alpha$.  But of course, in that case, $\setc_1$ is itself a smaller and valid combination. If some, but not all, the sets are almost surely nested, the natural way to merge them is to pick the smallest one of the nested ones, and combine it with the others via majority vote.

\subsection{Combining exchangeable uncertainty sets}\label{subsec:comb_exch}

In many practical applications, the independence of sets is often violated. Surprisingly, when $\setc_1, \dots, \setc_K$ are not independent, but are exchangeable, something better than a naive majority vote can be accomplished. To describe the method, denote the set~\eqref{eq:cm} as $\setc^M(1:K)$ to highlight that it is based on the majority vote of sets $\setc_1, \dots, \setc_K$. Now define
\[
\setc^E := \bigcap_{k=1}^K \setc^M(1:k),
\]
which can be equivalently represented as 
\begin{equation}\label{eq:setce}
   \setc^E = \left\{s \in \set: \frac{1}{k} \sum_{j=1}^k \ind\{s \in \setc_j\}>\frac{1}{2} \mathrm{~for~all~} k \le K\right\}. 
\end{equation}
Essentially, $\setc^E$ is formed by the intersection of sets obtained through sequential processing of the sets derived from the majority vote.
\begin{theorem}\label{thm:exch}
    If $\setc_1, \dots, \setc_K$ are $K \geq 2$ exchangeable sets having coverage $1-\alpha$, then $\setc^E$ is a $1-2\alpha$ uncertainty set, and it is never worse than majority vote ($\setc^E \subseteq \setc^M$).
\end{theorem}

The above result immediately implies that for multi-split conformal prediction, as studied in~\cite{solari2022}, one can obtain tighter prediction sets than their work without any additional assumptions; we will discuss this example in the following. We note that $\setc^M(1:2)$ is the intersection of $\setc^1$ and $\setc^2$, so we can omit $\setc^M(1:1) = \setc^1$ from the intersection defining $\setc^E$, to observe that $\setc^E = \bigcap_{k=2}^K \setc^M(1:k)$.

Since the set $\setc^E$ depends on the order of arrival of the different sets, one may worry that $\setc^E$ is less stable than $\setc^M$ (that treats the sets symmetrically). However, the exchangeable method allows the user to combine the obtained sets sequentially and to stop when the procedure stabilizes. This can be particularly advantageous, especially if the procedure to obtain the sets is computationally expensive.

Despite the previous result working only for exchangeable sets, it points out a simple way at improving majority vote for arbitrarily dependent sets: process them in a random order. To elaborate, let $\pi$ be a uniformly random permutation of $\{1,2,\dots,K\}$ that is independent of the $K$ sets, and define
\begin{equation}\label{eq:c^pi}
    \setc^\pi := \bigcap_{k=1}^K \setc^M(\pi(1):\pi(k)).
\end{equation}

Since $\setc^M(\pi(1):\pi(K)) = \setc^M(1:K)$, $\setc^\pi$ is also never worse than majority vote despite satisfying the same coverage guarantee:

\begin{corollary}\label{cor:arbit-exch}
    If $\setc_1, \dots, \setc_K$ are $K \geq 2$ arbitrarily dependent uncertainty sets having coverage $1-\alpha$, and $\pi$ is a uniformly random permutation independent of them, then $\setc^\pi$ is a $1-2\alpha$ uncertainty set, and it is never worse than majority vote ($\setc^\pi \subseteq \setc^M$).
\end{corollary}

The proof follows as a direct corollary of Theorem~\ref{thm:exch} by noting that the random permutation $\pi$ induces exchangeability of the sets (the joint distribution of every permutation of sets is the same, due to the random permutation). Of course, if the sets were already ``randomly labeled'' 1 to $K$ (for example, to make sure there was no special significance to the labels), then the aggregator does not need to perform an extra random permutation.

\subsection{Improving majority vote via random thresholding}
\label{sec:randomized_maj_vote}
Moving in a different direction below, we demonstrate that the majority vote can be improved with the aim of achieving a tighter set through the use of independent randomization, while maintaining the same coverage level.

Let $U$ be an independent random variable that is distributed uniformly on $[0,1]$.
We then define a new set $\setc^R$ as:
\begin{equation}
\label{eq:cr}
    \setc^R := \left\{s \in \set: \frac{1}{K}\sum_{k=1}^K \ind\{s \in \setc_k\} > \frac{1}{2} + U/2 \right\}.
\end{equation}
As a small variant, define
\begin{equation}
\label{eq:cu}
    \setc^U := \left\{s \in \set: \frac{1}{K}\sum_{k=1}^K \ind\{s \in \setc_k\} > U \right\}.
\end{equation}

\begin{theorem}
Let $\setc_1, \dots, \setc_K$ be $K \geq 2$ different sets satisfying the property \eqref{eq:coverage}. Then, the set $\setc^R$ has coverage at least $1-2\alpha$ and is never larger than $\setc^M$, while the set $\setc^U$ has coverage at least $1-\alpha$ and is never smaller than $\setc^R$.
\end{theorem}

The proof follows as a special case of the next subsection's result and is thus omitted. Even though $\setc^U$ does not improve on $\setc^M$, we include it since it involves random thresholding and delivers the same coverage as the input sets, a feature that we do not know how to obtain without randomization (unless in one of the cases described in Section~\ref{sec:independent_sets}).

There may of course be certain concerns with using randomization when a human is actively involved with data analysis, due to fears of a form of ``p-hacking'' caused by repeatedly running the above procedure many times and picking the most suitable merged set for the story they wish to tell. There is no fool-proof way to avoid this. We recommend using randomization in automated pipelines because it is more statistically efficient to do so, but we suggest using the non-randomized procedures when humans are involved.

\subsection{Weighted majority vote}
\label{sec:weighted_maj_vote}

It is not unusual for each interval to be assigned distinct ``weights'' (importances) in the voting procedure. This can occur, for instance, when prior studies empirically demonstrate that specific methods for constructing uncertainty sets consistently outperform others. Alternatively, a researcher might assign varying weights to the sets based on their own prior insights. 
Assume as before that the sets $\setc_1, \dots, \setc_K$ based on the observed data follow the property \eqref{eq:coverage}. In addition, let $w = (w_1, \dots, w_K)$ be a set of weights, such that 
\begin{gather}
\label{eq:weights}
    w_k \in [0, 1] \quad  \mathrm{and} \quad \sum_{k=1}^K w_k = 1.
\end{gather}
These weights can be interpreted as the aggregator's prior belief in the quality of the received sets. A higher weight signifies that we attribute greater importance to that specific interval. 
As before, let $U$ be an independent random variable that is distributed uniformly on $[0,1]$.
We then define a new set $\setc^W$ as:
\begin{equation}
\label{eq:cw}
    \setc^W := \left\{s \in \set: \sum_{k=1}^K w_k \ind\{s \in \setc_k\} > \frac{1}{2} + U/2 \right\}.
\end{equation}
\begin{theorem}
\label{th:CR}
    Let $\setc_1, \dots, \setc_K$ be $K \geq 2$ different sets satisfying property \eqref{eq:coverage}. Then, the set $\setc^W$ defined in \eqref{eq:cw} has coverage of at least $1-2\alpha$:
    \begin{equation}
    \label{eq: coverage_cr}
    \pr\big(c \in \setc^W\big) \geq 1-2\alpha.
    \end{equation}
    In addition, let $\nu(\setc^W)$ be the measure associated with the set $\setc^W$, then
    \begin{equation}
        \label{eq:length_cr}
        \nu(\setc^W) \leq 2 \sum_{k=1}^K w_k \nu(\setc_k).
    \end{equation}
\end{theorem}


If the weights are equal to $w_k = \frac{1}{K}$, for all $k=1, \dots, K$, then the set $\setc^W$ coincides with the set $\setc^R$ defined in \eqref{eq:cr} and it is a subset of that in \eqref{eq:cm}. This means that in the case of a democratic (equal-weighted) vote $\setc^R \subseteq \setc^M$, since $\setc^M$ is obtained by choosing $U=0$. Furthermore, \eqref{eq:length_cr} says that the measure of the set obtained using the weighted majority method cannot be more than twice the average measure obtained by randomly selecting one of the intervals with probabilities proportional to $w$. When there are only two sets and randomization is avoided, the weighted majority vote set will correspond to the set with the greater weight. 
In Appendix \ref{sec:inf_number} we extend the result to the case where the number of the sets can be uncountable, while applications of the proposed methods are reported in Appendix \ref{sec:add_appl}.

\subsection{Different coverage levels and asymptotic coverage}

It may happen that the property \eqref{eq:coverage} is not met, meaning that the sets may have different coverage. The agents may intend to provide a $1-\alpha$ confidence set, but may unintentionally overcover or undercover, or some agents could be malevolent. As examples of the former case, we know that under regularity conditions confidence intervals constructed using likelihood methods have an asymptotic coverage of level $1-\alpha$ \cite[chap.3]{pace1997}, but the asymptotics may not yet have kicked in or the regularity conditions may not hold. Another example appears in the conformal prediction framework when the exchangeability assumption is not satisfied but we do not know the amount of deviation from exchangeability, as considered in \cite{barber2023}. As another example in conformal prediction, the jackknife+ method by~\cite{barber2021} when run at level $\alpha$ may deliver coverage anywhere between $1-2\alpha$ and 1, even under exchangeability. As a last example, a Bayesian agent may provide a credible interval, which may not be a valid confidence set in the frequentist sense. The set in \eqref{eq:cw} still delivers a sensible guarantee.

\begin{proposition}
    If $\setc_1, \dots, \setc_K$ are $K \geq 2$ different sets having coverage $1-\alpha_1, \dots, 1-\alpha_K$ (possibly unknown), then the set $\setc^W$  defined  in \eqref{eq:cw} has coverage
    \begin{equation*}
        \pr(c \in \setc^W) \geq 1-2\sum_{k=1}^K w_k \alpha_k.
    \end{equation*}
    In particular, this implies that the majority vote of asymptotic $(1-\alpha)$ intervals has asymptotic coverage at least $(1-2\alpha)$.
\end{proposition}

The proof is identical to that of Theorem~\ref{th:CR}, with the exception that the expected value for the variables $\phi_k$ is equal to $\alpha_k$, and is thus omitted. If the $\alpha_k$ levels are known (which they may not be, unless the agents report it and are accurate), and if one in particular wishes to achieve a target level $1-\alpha$, then it is always possible to find weights $(w_1, \dots, w_K)$ that achieve this as long as $\alpha/2$ is in the convex hull of $(\alpha_1, \dots, \alpha_K)$.


\bigskip

Since it is desirable to have as small an interval as possible if coverage \eqref{eq:coverage} is respected, we would like to assign a higher weight to intervals of smaller size. However, the weights must be assigned before seeing the sets.

\subsection{Sequentially combining uncertainty sets}\label{subsec:sequential}

Here, we show simple extensions of the results to two different sequential settings: 
\subsubsection{Sequential data:}   Imagine that we observe data sequentially one at a time $Z_1,Z_2,\dots$ and wish to estimate some parameter $c$ with increasing accuracy as we observe more samples, and wish to continuously monitor the resulting confidence intervals as they get tighter over time. A $(1-\alpha)$-confidence sequence $(\setc^{(t)})_{t \geq 1}$ for  $c$ is a time-uniform confidence interval:
\[
\pr(\forall t\geq 1: c \in \setc^{(t)}) \geq 1-\alpha.
\]
Here $t$ indexes sample size that is used to calculate a single agent's confidence interval $\setc^{(t)}$ (meaning that $\setc^{(t)}$ is based on $(Z_1,\dots,Z_t)$). 
Suppose now that we have $K$ different \emph{confidence sequences} for a parameter that need to be combined into a single confidence sequence.
For this setting we show a simple result:
\begin{proposition}
    Given $K$ different $1-\alpha$ confidence sequences for the same parameter that are being tracked in parallel, their majority vote set is a $1-2\alpha$ confidence sequence.
\end{proposition}

It may not be initially apparent how to deal with the time-uniformity in the definitions. The proof proceeds by first observing that an equivalent definition of a confidence sequence is a confidence interval that is valid at any arbitrary stopping time $\tau$ (here the underlying filtration is implicitly that generated by the data itself). \cite{howard2021} proved that $(\setc_k^{(t)})_{t \geq 1}$ is a confidence sequence if and only if for every stopping time $\tau$,
\(
\pr(c \in \setc_k^{(\tau)}) \geq 1-\alpha,
\)
for all $k=1,\dots,K$. Now, the proof follows by applying our earlier results.

\subsubsection{Sequential set arrival:}
Now, let  the data be fixed, but consider the setting where an unknown number of confidence sets arrive one at a time in a random order, and need to be combined on the fly. Now, we propose to simply take a majority vote of the sequences we have seen thus far. Borrowing terminology from earlier, denote 
\begin{equation}\label{eq:sequential_mer}
\setc^E(1:t) := \bigcap_{i=1}^t \setc^M(1:i).    
\end{equation}

We claim that the above sequence of sets  is actually a $1-2\alpha$ confidence \emph{sequence} for $c$:
\begin{theorem}
\label{thm:exch2}
    Given an exchangeable sequence of confidence sets $\setc_1, \setc_2,\dots$ (or confidence sets arriving in a uniformly random order), the sequence of sets formed by their ``running majority vote'' $(\setc^E(1:t))_{t \geq 1}$ is a $1-2\alpha$ confidence sequence:
    \[
    \pr\left(\exists t \geq 1: c \notin \setc^{E}(1:t)\right)\leq 2\alpha.
    \]
\end{theorem}

This result will be particularly useful in the next section on derandomization.

\section{Derandomizing statistical procedures}
\label{sec:derand}
One application of the presented bounds is in the derandomization of existing randomized methods that are based in some way on data splitting. We explore different such methods here, such as the median-of-means~\citep{devroye2016sub} and HulC~\citep{kuchibhotla2021hulc}. The first of these produces a point estimator, while the second produces an interval, but both can be derandomized using the same set of ideas. 

Since the paper has so far focused on uncertainty sets, we first describe a general result that derandomizes point estimators ``directly''.

\begin{theorem}
\label{thm:point}
    Suppose $\hat \theta_1,\dots,\hat \theta_K$ are $K$ univariate point estimators of $\theta$ that are each built using $n$ data points and satisfy a high probability concentration bound:
    \[ 
    \pr( |\hat \theta_k - \theta| \leq w(n,\alpha)) \geq 1-\alpha,
    \] 
    for some function $w$. Then, their median $\theta_{(\lceil K/2 \rceil)}$ satisfies
    \[ 
    \pr(|\hat \theta_{(\lceil K/2 \rceil)} - \theta| \leq w(n,\alpha)) \geq 1-2\alpha. 
    \] 
    Further, if $\hat \theta_1,\hat \theta_2,\dots,\hat \theta_K, \hat\theta_{K+1} \dots$ are exchangeable, then
    \[
    \pr(\forall K \geq 1: |\hat \theta_{(\lceil K/2 \rceil)} - \theta| \leq w(n,\alpha)) \geq 1-2\alpha. 
    \]
    An analogous statement also holds if we instead had $\pr(\ell(n,\alpha) \leq \hat \theta_k - \theta \leq u(n,\alpha)) \geq 1-\alpha$, meaning that the two tails had different behaviors.
\end{theorem}

In the particular context of ``double machine learning'', \citet[Corollary 3.3]{chernozhukov2018double} also proposes repeating their sample-splitting based procedure several times and taking either the median or the mean of the resulting estimates, but their arguments justifying this choice are asymptotic. Further, these asymptotic arguments do not distinguish between the median and mean combination rules, but the authors recommend the median rule because it is more robust to outliers. Our justification above is nonasymptotic and rather different in flavor, and it applies in broader situations, beyond the ones considered in the above work.
Indeed, derandomization can be applied in other contexts involving data splitting; some examples are \cite{decrouez2014, banerjee2019, cholaquidis2024gros, kim2024}.

\subsection{The Median-of-Median-of-Means (MoMoM) procedure}

The aim of the celebrated median-of-means procedure is to produce estimators of the mean that have subGaussian tails, despite the underlying unbounded data having only a finite variance. It stems back to, at least, a book by~\cite{nemirovskij1983problem}, but some modern references include~\cite{devroye2016sub} and \cite{lugosi2019mean}.

The setup assumes $n$ iid data points from an unknown univariate distribution $P$ with unknown finite variance, whose mean $\mu$ we seek to estimate. The method works by randomly dividing the $n$ points into $B$ buckets of roughly $n/B$ points each. One takes the mean within each bucket, and then calculates the median $\hat \mu^{\mom}$ of those numbers.

The method does not produce confidence intervals for the mean; indeed, one needs to know a bound on the variance $\sigma^2$ for any nonasymptotic CI to exist. But the point estimator $\hat \mu^{\mom}$ satisfies the following subGaussian tail bound: for any $t \geq 0$,
\[
\pr(|\hat \mu^{\mom} - \mu| \geq C \sigma \sqrt{t/n}) \leq 2 e^{-t},
\]
for a constant $C = \sqrt{\pi} + o(1)$, where the $o(1)$ term vanishes as $B, n/B \to \infty$.
This is a much faster rate than the $1/t^2$ on the right hand side obtained for the sample mean via Chebyshev's inequality.
However, one drawback of the method is that it is randomized, meaning that it depends on the random split of the data into buckets. \cite{minsker2022u} proposed a derandomized variant that even improves the constant $C$ to $\sqrt{2} + o(1)$. However, it is computationally intensive: it involves taking the median of \emph{all possible sets} of size $n/B$.

Here, we point out that Theorem~\ref{thm:point} yields a simple way to derandomize $\hat \mu^{\mom}$. We simply repeat the median-of-means procedure independently $K$ times to obtain exchangeable estimators $\{\hat \mu^{\mom}_k\}_{k=1}^K$, and then report the ``median of median of means'' (MoMoM):
\[
\hat \mu^{\momom}_K := \hat \mu^{\mom}_{\lceil K/2 \rceil}.
\] 

Clearly, $\{\hat \mu^{\mom}_k\}_{k=1}^K$ form an exchangeable set, whose empirical distribution will stabilize as $K \to \infty$, and thus whose median will also be stabilize for large $K$. 

\begin{corollary}
    Under the setup described above, the \emph{median of median of means} satisfies a nearly identical subGaussian behavior as the original:
\[
\pr\left(|\hat \mu^{\momom}_K - \mu| \geq C \sigma \sqrt{t/n}\right) \leq 4 e^{-t},
\]
for any $t \geq 0$ and $C=\sqrt{\pi} + o(1)$, as before. In fact,
\[
\pr\left(\exists K \geq 1: |\hat \mu^{\momom}_K - \mu| \geq C \sigma \sqrt{t/n}\right) \leq 4 e^{-t}.
\]
\end{corollary}
The proof is an immediate consequence of Theorem~\ref{thm:point} and is thus omitted. The simultaneity over $K$ allows us to choose any sequence of constants $k_1 \leq k_2 \leq \dots$ and produce the sequence of estimators $\hat \mu^{\momom}_{k_1}, \hat \mu^{\momom}_{k_2}, \dots$, and plot this sequence as it is being generated. We can then stop the derandomization  whenever the plot appears to stabilize, with the guarantee now holding for whatever data-dependent random $K$ we stopped at. Clearly, one can use the same technique to derandomize other random estimators by taking their median while ``importing'' their nonasymptotic bounds.



\subsubsection{Simulation study.} As previously mentioned, a potential method to derandomize the MoM is to repeat the procedure $K$ times and subsequently take the median of the estimators. The natural question that arises is: how large should $K$ be in practice to achieve a derandomized result? Indeed, considering computational and time resources, one would prefer a small value for $K$. We conducted a simulation study to study the impact of $K$. 

\begin{figure*}
    \centering
    \includegraphics[width=0.85\textwidth]{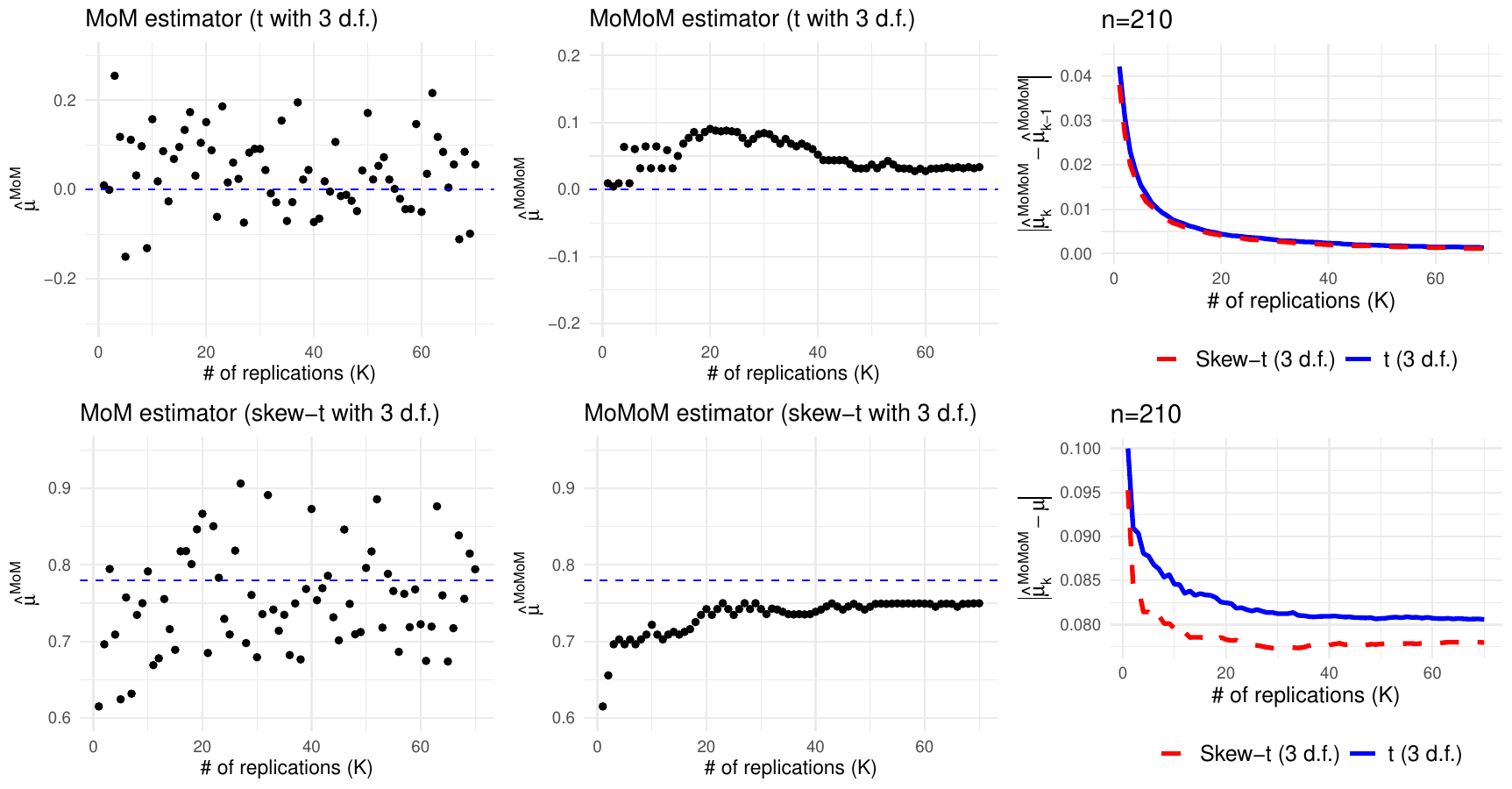}
    \caption{First column: MoM estimator obtained during various replications in a \emph{a single run of the procedure} using data generated from the t-distribution (above) and the skew-t distribution (below). The dashed line is the true mean $\mu$. Second column: MoMoM estimator obtained during the replications; both series stabilize for $k>40$. Third column: average over 1000 runs of the absolute difference $|\hat\mu^\momom_k-\hat\mu^\momom_{k-1}|$ against $k$ (above) and average over 1000 runs of the absolute difference $|\hat\mu^\momom_k-\mu|$ against $K$. Note that the third column is only available to us in the simulation, but plots like the second column can be constructed on the fly as $K$ increases, and it can be adaptively tracked and stopped, while retaining the same statistical guarantee at the stopped $K$.}
    \label{fig:momom_diff}
\end{figure*}

We simulated observations from a standard Student's t-distribution with 3 degrees of freedom and from a skew-t distribution with 3 degrees of freedom and 1 as a skewness parameter~\citep{azzalini2003}. The number of batches is equal to $B=21$ and the number of data points generated are $n=210$. For each iteration, we computed $\hat\mu^\momom_1, \dots, \hat\mu^\momom_K$ with $K = 70$ and considered the absolute value of the difference $\hat\mu^\momom_k-\hat\mu^\momom_{k-1}$ as a measure of stability. In the first two columns of Figure~\ref{fig:momom_diff}, we present two examples of the procedure, where it can be observed that in the first case, $\hat\mu^\momom$ tends to stabilize after 50 iterations, while in the second case, where the data are generated from a skew-t distribution, 25 iterations are sufficient.
In the third column of Figure~\ref{fig:momom_diff}, we report the empirical average over 1000 replications of $|\hat\mu^\momom_k-\hat\mu^\momom_{k-1}|$. We observe an inflection point around 25 for both distributions, and it seems that once this threshold is reached, the gain becomes negligible.

\subsection{Derandomizing HulC}

Recently, \cite{kuchibhotla2021hulc} proposed a rather general inference procedure for deriving confidence intervals for a target functional $\theta$, using any point estimator $\hat \theta$. Their HulC procedure starts off in a similar fashion to the median-of-means: it divides the $n$ iid data points into $B$ buckets of $n/B$, and computes the point estimator $\hat \theta_b$ in each bucket (this could be the sample mean, as in the previous subsection). It then reports a confidence interval by using certain quantiles of $\{\hat \theta_b\}_{b=1}^B$ (as opposed to a point estimator given by their median). 
The miscoverage rate is not exactly $\alpha$, but is determined by the ``median bias'' of the underlying estimator, a quantity that must vanish asymptotically for any nonasymptotic confidence interval to exist. It is important to remark that the final output of the HulC is a confidence set for the functional of interest. Clearly, given a level $\alpha$, one could test the hypothesis that $\theta$ equals a certain value simply by checking if that value falls within the level $1-\alpha$ interval. But, calculating the p-value for a given hypothesis is not straightforward, as there is no guarantee that the sets produced with different coverage levels are nested as $\alpha$ decreases due to the randomness of the procedure (this is due to the fact that the number of buckets depends on $\alpha$). This makes our method particularly useful in this setting, since it simply requires the sets and not the underlying p-values. 

As with MoM, the HulC confidence interval depends on the random split of the data. Our majority vote procedure effectively derandomizes the procedure while retaining the optimal rate of shrinkage (in settings where HulC possesses this property).

We remark that the similarity between the first steps of HulC and MoM has not been previously explicitly noted, but it is interesting because HulC effectively provides a confidence interval for the MoM procedure. The coverage is not exactly $1-\alpha$ due to the unknown median bias, but whatever the resulting coverage is, a similar coverage is retained by our various majority vote sets. (Indeed a nonasymptotically valid confidence interval is impossible without any apriori bound on the variance.)

\subsubsection{Simulation study.} We conducted a simulation study wherein we generated data from a Student's t-distribution with 3 degrees of freedom. The HulC method was used to obtain a confidence interval at a level $1-\alpha$, and the Median of Means (MoM) was used as the estimator in this context. First, we define the median bias of the estimator $\hat \mu^\mom$ as
\[
\begin{split}
    &\text{Med-Bias}_{\mu^*}(\hat \mu^\mom) =\left(\frac{1}{2} - \min\left\{\pr(\hat \mu^\mom \geq \mu^*), \pr(\hat \mu^\mom \leq \mu^*)\right\}\right)_+,
\end{split}
\]
where $\mu^*$ denotes the true mean of the distribution. In this case, the empirical means used as estimators in each batch are median-unbiased, i.e.\
\(
    \pr(\hat \mu_j \geq \mu^*) =  \pr(\hat \mu_j \leq \mu^*) = \frac{1}{2},
\)
since the data are independent and identically distributed from a symmetric location family \citep{kuchibhotla2023median}. We define $B_1$ and $B_2$ as the number of splits used for the MoM and HulC procedure, respectively. If $B_1$ is odd, then $\pr(\hat \mu^\mom \geq \mu^*)$ implies that at least $\lceil B_1/2 \rceil$ of the estimators $\hat\mu_j$ are greater than $\mu^*$. Since the $\hat \mu_j$ are independent, this corresponds to the probability that a $\text{Binom}(B_1,1/2)$ is greater than $\lceil B_1/2 \rceil-1$ which is equal to half. The same result is obtained with $\pr(\hat \mu^\mom \leq \mu^*)$. This implies that if $B_1$ is odd, then $\hat\mu^\momom$ is a median unbiased estimator for the mean. The number of data splits $B_2$ to reach an interval with miscoverage equal to $\alpha = 0.05$ is 6 (see eq. 4 in \cite{kuchibhotla2021hulc}) while we fix $n=210$ and $B_1=7$. Since intervals are constructed using different data splits one can use $\setc^E$ to build a valid $1-2\alpha$. From Figure~\ref{fig:Mom+Hulc}, we see that the intervals $\setc^M$ are stable if the number of replications is greater than 5, and $\setc^E$ is simply given by the running intersection of these intervals. In the third column of Figure~\ref{fig:Mom+Hulc}, we see that $\setc^M$ tends to produce intervals with an higher coverage, while $\setc^E$ stabilizes around $1-2\alpha$. 


\begin{figure*}
    \centering
    \includegraphics[width=1\textwidth]{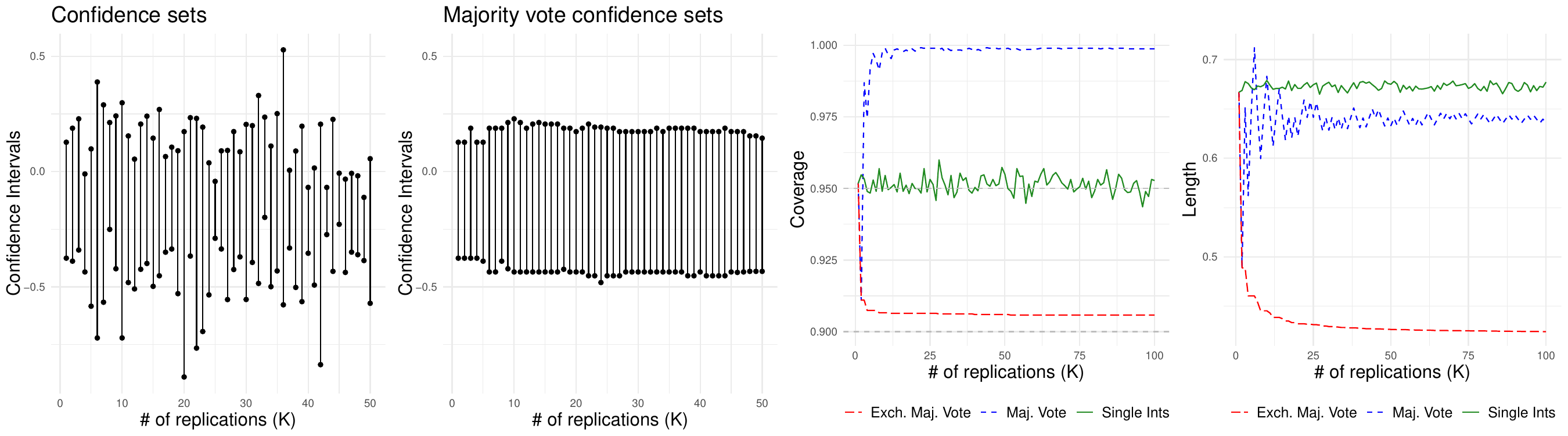}
    \caption{First two columns: example of \emph{a single run of the procedure}  using 210 observations generated from the t-distribution. Left: confidence intervals obtained for different random splits, Right: the merged sets. Last two columns: average over $5 000$ runs of the empirical coverage (left) and length (right) of $\setc^M$ and $\setc^E$ against $K$.}
    \label{fig:Mom+Hulc}
\end{figure*}

\section{Applications to conformal prediction}
\label{sec:conformal}
Conformal prediction is a popular framework to obtain prediction sets with a prespecified level of (marginal) coverage and without assuming correctness of an underlying model or strong distributional assumptions; see \cite{vovk2005,angelopoulos2021}. This method is now widely employed to obtain prediction sets for ``black box" algorithms.

Suppose we have independent and identically distributed random vectors $Z_i=(X_i, Y_i),\, i=1, \dots, n,$ from some unknown distribution $P_{XY}$ on the sample space $\mathcal{X} \times \real$, where $\mathcal{X}$ represents the space of covariates. In addition, suppose that $K$ different agents construct $K$ different conformal prediction sets $\setc_1(x), \dots, \setc_K(x)$ with level $1-\alpha$ based on the observed training data $z_i = (x_i, y_i), i=1,\dots, n,$ and a test point $x \in \mathcal{X}$. By definition, a conformal prediction interval with level $1-\alpha$ has the following property:
\begin{equation}
\pr\left(Y_{n+1} \in \setc_k(X_{n+1})\right) \geq 1 - \alpha, \quad k = 1, \dots, K,
\end{equation}
where $\alpha \in (0,1)$ is a user-chosen error rate. It is important to highlight that this form of guarantee is marginal, indicating that the coverage is calculated over a random draw of the training data and the test point. The $K$ different intervals can differ due to the algorithm used to obtain the predictions or the variant of conformal prediction employed \citep{lei2018, romano2019, barber2021}.

Recently, \cite{fan2023} have proposed a method to merge prediction intervals (or bands) with the aim of minimizing the average width of the interval. This method employs linear programming and is grounded in the assumption that the response can be expressed as the sum of a mean function plus a heteroskedastic error. In the context of combining conformal prediction sets from $K$ different algorithms for a single data split, another method is introduced by \cite{yang2021}. In particular, starting from $(1-\alpha)$-prediction intervals, they prove that the training conditional validity obtained by their method differs from $1-\alpha$ by a constant that depends on the number of algorithms and the number of points in the calibration set. Our black-box setting and aggregation method are both quite different from theirs and they can be considered as an extension of the method introduced in \cite{solari2022}. Theorems~\ref{th:CM} and~\ref{th:CR} are specialized (in the conformal case) to obtain the following result:
\begin{corollary}
\label{cor:conformal}
    Let $\setc_1(x), \dots, \setc_K(x)$ be $K \geq 2$ different conformal prediction intervals obtained using observations $(x_1, y_1), \dots, (x_n, y_n)$, $x \in \mathcal{X}$ and $w=(w_1, \dots, w_k)$ defined as in \eqref{eq:weights}. Then,
    \begin{gather*}
        \setc^M(x) = \left\{ y \in \real: \frac{1}{K} \sum_{k=1}^K \ind\{y \in \setc_k(x) \} > \frac{1}{2}\right\},\\
        \setc^W(x) = \left\{ y \in \real: \sum_{k=1}^K w_k \ind\{y \in \setc_k(x) \} > 1/2 + U/2 \right\}, 
    \end{gather*}
    where $U \sim \mathrm{Unif}(0, 1)$, are valid conformal prediction sets with level $1-2\alpha$.
\end{corollary}

Suppose that we have constructed $K$ arbitrarily dependent prediction sets using conformal prediction, then, according to Corollary~\ref{cor:conformal}, we can merge the sets using a majority vote procedure while maintaining an appropriate level of coverage. If the conformal method used ensures an upper bound on coverage, then  \eqref{eq:upper_bound} still holds, with the difference that $\alpha$ is replaced by this upper limit. As a matter of fact, methods such as split or full conformal, under weak conditions, exhibit coverage that is practically equal to the pre-specified level. 

As explained in Section~\ref{subsec:comb_exch}, if the sets are exchangeable then it is possible to obtain better results than using a simple majority vote procedure. Specifically, Theorem~\ref{thm:exch} can be specialized in the context of conformal prediction.
\begin{corollary}
    Let $\setc_1(x), \dots, \setc_K(x)$ be $K \geq 2$ exchangeable conformal prediction intervals having coverage $1-\alpha$, then
    \begin{equation*}
    \begin{split}
        \setc^E(x) =\bigg\{&y \in \mathbb{R}: \frac{1}{k} \sum_{j=1}^k \ind\{y \in \setc_j(x)\} > \frac{1}{2} \mathrm{~}\forall k \le K \bigg\},
    \end{split}
    \end{equation*}
    is a valid $1-2\alpha$ conformal prediction set.
\end{corollary}
If the sets are non-exchangeable then exchangeability can be achieved through a random permutation $\pi$ of the indexes $\{1, \dots, K\}$, in order to obtain the set $\setc^\pi(x)$ described in \eqref{eq:c^pi}.

In the following, we will study the properties of the method through a simulation study and an application to real data. One consistent phenomenon that we seem to observe empirically is that $\setc^M$ actually has coverage $1-\alpha$ (better than $1-2\alpha$ as promised by the theorem), and the smaller sets $\setc^R$ and $\setc^\pi$ have coverage between $1-\alpha$ and $1-2\alpha$.

\subsection{Simulations}\label{sec:conf_sim}
We carried out a simulation study in order to investigate the performance of our proposed methods. We apply the majority vote procedure on simulated high-dimensional data with $n=100$ observations and $p=120$ regressors. Specifically, we simulate the design matrix $X_{n \times p}$, where each column is independent of the others and contains standard normal entries. The outcome vector is equal to $y = X\beta + \varepsilon$, where $\beta$ is a sparse vector with only the first $m=10$ elements different from 0 (generated independently from a $\mathcal{N}(0, 4)$) while $\varepsilon \sim \mathcal{N}_n(0, I_n)$. A test point $(x_{n+1}, y_{n+1})$ is generated with the same data-generating mechanism. At each iteration we estimate the regression function using the lasso algorithm \citep{tibshirani1996} with penalty parameter $\lambda$ varying over a fixed sequence of values $K=20$ and then construct a conformal prediction interval for each $\lambda$ in $x_{n+1}$ using the split conformal method presented in package R \texttt{ConformalInference}. We run  $10\,000$ iterations at $\alpha=0.05$. 

We then merge the $K$ different sets using the method described in Corollary~\ref{cor:conformal} with $w_k=\frac{1}{K}, k=1,\dots, K$ (an additional simulation to study the impact of $w$ is reported in Appendix~\ref{sec:add_sims_w}). These weights can be interpreted, from a Bayesian perspective, as a discrete uniform prior on $\lambda$. From an alternative perspective, each agent represents a value of the penalty parameter, and the \emph{aggregator} equally weighs the various intervals constructed by the various agents. An example of the result is shown in Figure~\ref{fig:lasso_sims}. The empirical coverages of the intervals $\setc^M(x)$ and $\setc^R(x)$ are $\frac{1}{B}\sum_{b=1}^B \ind\{y_{n+1}^b \in \setc^M_b(x_{n+1}) \}=0.97$ and $\frac{1}{B}\sum_{b=1}^B \ind\{y_{n+1}^b \in \setc^R_b(x_{n+1}) \}=0.92$. By definition, the second method produces narrower intervals while maintaining the coverage level $1-2\alpha$. As explained in the previous sections, by inducing exchangeability through permutation, it may be possible to enhance the majority vote results. In fact, the empirical coverage of the sets $\setc^\pi$ is $0.93$ while the sets are smaller than the ones produced by the simple majority vote. Furthermore, we tested the sets $\setc^U(x)$ defined in \eqref{eq:cu} and obtained an empirical coverage equal to $0.96$, which is very close to the nominal level $1-\alpha$. In all five cases, the occurrence of obtaining an union of intervals as output is very low, specifically less than 1\% of the iterations.

\begin{figure}
    \centering
    \includegraphics[width=0.8\textwidth]{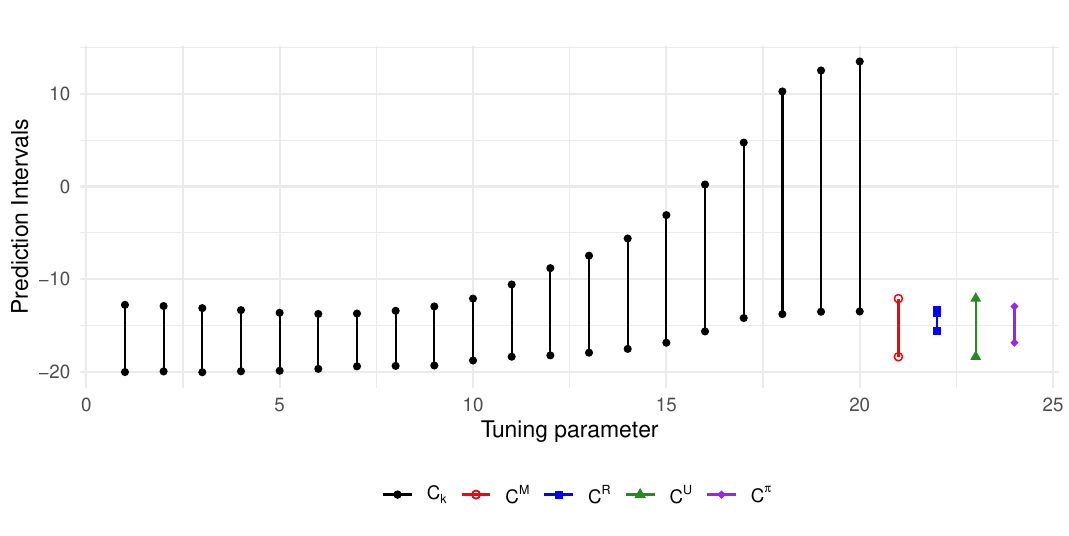}
    \caption{Intervals obtained using different values of $\lambda$, $\setc^M(x)$, $\setc^R(x), \setc^U(x)$ and $
    \setc^\pi(x)$. For standardization, the value of $U$ in the randomized thresholds is set to $1/2$. The smallest set $\setc^R$. Since $U=1/2$, the sets $\setc^M$ and $\setc^U$ coincides.}
    \label{fig:lasso_sims}
\end{figure}

\subsection{Real data example}
We used the proposed methods in a real dataset regarding Parkinson's disease \citep{parkinson_data}. The goal is to predict the total UPDRS (Unified Parkinson's Disease Rating Scale) score using a range of biomedical voice measurements from people suffering early-stage Parkinson's disease. We used split conformal prediction and $K=4$ different algorithms (linear model, lasso, random forest, neural net) to obtain the conformal prediction sets. In particular, we choose $n=5000$ random observations to construct our intervals and the others $n_0=875$ observations as test points. Prior weights also in this case are uniform over the $K$ models that represent the different agents, so a priori all methods are of the same importance. If previous studies had been carried out, one could, for example, put more weight on methods with better performance. Otherwise, one can assign a higher weight to more flexible algorithms such as random forest or neural net. 

The results are reported in Table~\ref{tab:parkinson} where it is possible to note that all merging procedures obtain good results in terms of length and coverage. In addition, also the randomized vote obtains good results in terms of coverage, with an empirical length that is slightly larger than the one obtained by the neural net. The percentage of times that a union of intervals or an empty set is outputted is nearly zero for all methods. In this situation, the intervals produced by the random forest outperform the others in terms of size of the sets; as a consequence, one may wish to put more weight into the method, which results in smaller intervals on average.

\begin{table*}[h!]
\caption{Empirical coverage and length of the methods for the Parkinson’s dataset.}
\label{tab:parkinson}
\centering
\begin{tabular}{ccccccccc}
  \hline
  Methods & LM & Lasso & RF & NN & $\setc^M$ & $\setc^R$ & $\setc^U$ & $\setc^\pi$ \\ 
  \hline
  Coverage & 0.958 & 0.960 & 0.949 & 0.961 & 0.951 & 0.923 & 0.961 & 0.918  \\ 
  Width & 40.143 & 40.150 & 13.286 & 32.533 & 29.508 & 20.620 & 32.544 & 20.710 \\ 
   \hline
\end{tabular}
\end{table*}

\subsection{Multi-split conformal inference}

The \emph{full} conformal prediction method originally introduced by \cite{vovk2005}, exhibits good properties in terms of coverage and size of the prediction set; however, these are counterbalanced by a notable computational cost, which makes its practical application challenging. To address this issue, a potential solution is the adoption of \emph{split} conformal prediction \citep{papadopoulos2002,lei2018}, which involves the use of a random
split of the data into two parts. Although this variant proves to be highly computational efficient, it introduces an additional layer of randomness that stems from the randomness of the data split. In addition, only a fraction of the data is used to train the model, leading to a reduced efficiency of the prediction set in terms of width. Several works aim to mitigate this problem by proposing various ways to combine the intervals obtained from different splits \citep{solari2022, barber2021, vovk2015}.

In particular, the method introduced in \cite{solari2022} involves the construction of $K$ distinct intervals of level $1-\alpha(1-\tau)$, each originating from a different random split. Subsequently, these intervals are merged using the mechanism described in \eqref{eq:c^tau}.  Although the final interval achieves a coverage of $1-\alpha$, it is possible to enhance the procedure by exploiting the exchangeability of sets and using the results introduced in Section~\ref{subsec:comb_exch} and Section~\ref{subsec:sequential} with the threshold set to $\tau$. A simple possible solution is to fix $K$ in advance, construct the $K$ different sets in parallel, and subsequently merging them using the set $\setc^E$ described in~\eqref{eq:setce} (with a threshold different than $1/2$ and set to $\tau$). Another method is to not fix $K$ in advance and instead merge the sets online, through $C^E(1:t)$ described in~\eqref{eq:sequential_mer} and with a threshold equal to $\tau$ rather than half. In this case, the coverage $1-\alpha$ is guaranteed uniformly over the number of splits.

\begin{figure*}
    \centering
    \includegraphics[width=1\textwidth]{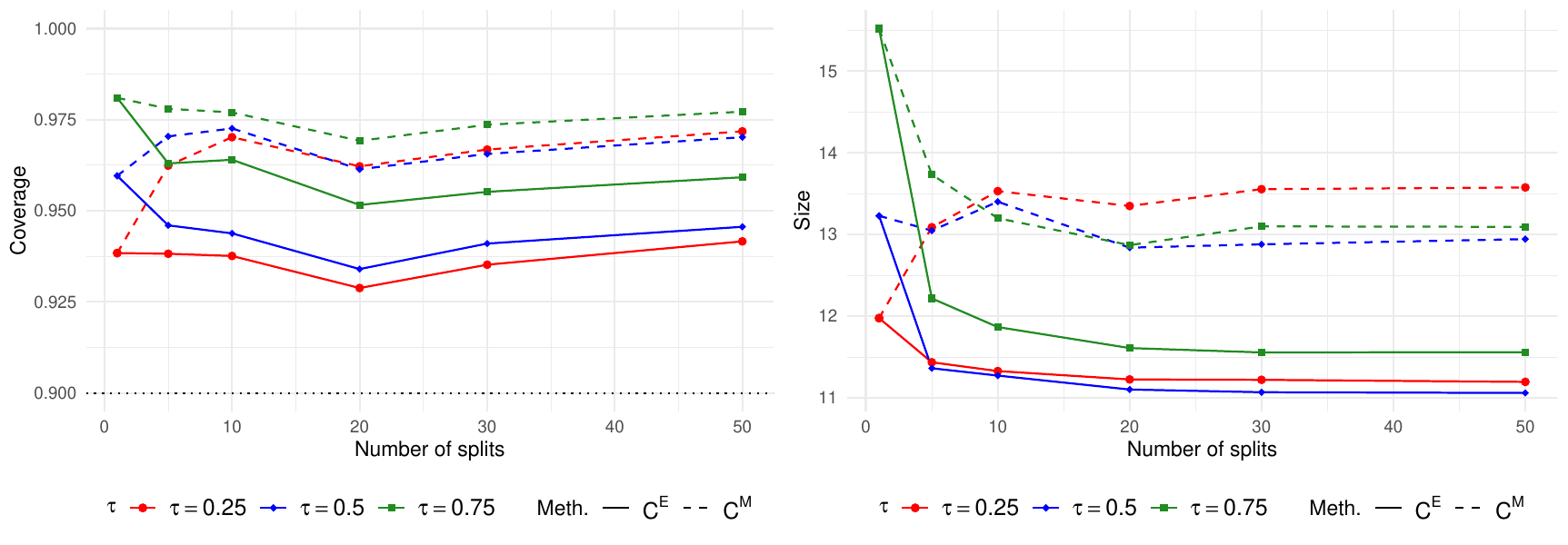}
    \caption{Comparison between the \emph{simple} majority vote  ($\setc^M$) and exchangeable majority vote ($\setc^E$) with different thresholds $\tau$. The size of the $\setc^E$ is significantly smaller than  $\setc^M$.} 
    \label{fig:multisplit}
\end{figure*}

We conducted a simulation study in which the data were generated using the same generative mechanism described in Section~\ref{sec:conf_sim}. Also in this case, the algorithm employed was Lasso with the penalty parameter set to one, while the intervals varied due to the data split on which they were constructed. 
The number of splits utilized was set at $\{1, 5, 10, 20, 30, 50\}$, while for the parameter $\tau$, three values were selected, namely $\{0.25, 0.50, 0.75\}$. The error rate $\alpha$ is set to $0.1$ and the procedure was repeated $5000$ times.

As evident from Figure~\ref{fig:multisplit}, the methods lead to a greater coverage than the specified level $1-\alpha$. However, the coverage is higher for the \emph{simple} majority vote compared to that based on the exchangeability of the sets. In addition, the size of the resulting sets is significantly smaller when using the method based on the exchangeability of sets. The number of splits does not appear to have a significant impact on the size of the sets. Once more, $\tau=0.5$ seems to offer good performances in terms of width.

\section{Summary}
Our paper presents a novel method to address the question of merging dependent confidence sets in an efficient manner, where efficiency is measured in both coverage and size. Our approach can be seen as the confidence interval analog of the results of \cite{ruger1978} which combines p-values through quantiles. The inclusion of a vector of weights allows the incorporation of prior information on the reliability of different methods. Additionally, the randomized version yields better results in terms of both coverage and width of the intervals, without altering the theoretical properties of the method. The improvements for exchangeable sets are particularly striking and are achieved using recently proposed extensions of Markov's inequality \citep{ramdas2023}.

The proposed method can be used to derandomize statistical procedures that are based on data-splitting. In both real and simulated examples, the method achieves good results in terms of coverage and width. The method has been extended in the Appendix~\ref{sec:conformal-risk-control} to sets with a conformal risk guarantee, introduced in \cite{angelopoulos2022}, allowing the extension of the results to different loss functions beyond miscoverage. 

The method is versatile and is clearly applicable in more scenarios than we have explored here.  We hope the community will explore such applications in future work.

\subsection*{Acknowledgments}
AR acknowledges funding from NSF grant IIS-2229881. The authors are also very thankful to Anasatasios Angelopoulos for an important pointer to related work, to Arun Kuchibhotla for a useful reformulation of one of our randomized methods, and Ziyu Xu for insights related to the measure of the majority vote set.

\bibliography{biblio}

\appendix
\section{Proofs of the results}
\label{sec:app_proof}

\begin{proof}[Proof of Theorem~\ref{th:up_bound}] Let
$r:=\left \lceil{\frac{K}{2}}\right\rceil$, $\phi_k=\phi_k(Z,c)=\ind\{c \notin \setc_k\}$ be a Bernoulli random variable such that $\Ev[\phi_k] = \alpha$, $k=1,\dots,K$, and $S_K = \sum_{k=1}^K \phi_k$ taking values in $\{0, 1, \dots, K\}$. By definition, we know that 
\[
\Ev[S_K] = \sum_{k=1}^K \Ev[\phi_k] = K\alpha.
\]
Let us define $\rho_j = \pr(S_K=j)$. Now we can write
\[
\begin{split}
    \Ev[S_K] &= \sum_{j=0}^K j \rho_j = \sum_{j=0}^{r-1} j\rho_j + \sum_{j=r}^K j\rho_j \pm (r-1)\sum_{j=0}^{r-1} \rho_j \pm K\sum_{j=r}^K \rho_j\\
    &= (r-1) \sum_{j=0}^{r-1} \rho_j + K \sum_{j=r}^K \rho_j - \sum_{j=0}^{r-1}(r-1-j)\rho_j - \sum_{j=r}^K (K-j)\rho_j\\
    &= (r-1)\left(1 - \pr\left(S_K \geq r \right) \right) + K \pr\left(S_K \geq r \right) - m.
\end{split}
\]
Since $m \geq 0$, then
\[
K\alpha \leq (r-1)\left(1 - \pr\left(S_K \geq \bigg\lceil\frac{K}{2}\bigg\rceil \right) \right) + K \pr\left(S_K \geq \bigg\lceil\frac{K}{2}\bigg\rceil \right) \implies \pr\left(S_K \geq \bigg\lceil\frac{K}{2}\bigg\rceil \right) \geq \frac{K\alpha - r + 1}{K - r + 1}
\]
From Theorem~\ref{th:CM} we know that 
\[
\pr(c \in \setc^M) = 1 - \pr(c \notin \setc^M) = 1 - \pr\left(S_K \geq \left\lceil \frac{K}{2} \right\rceil \right) \leq 1 - \frac{K\alpha - r + 1}{K - r + 1},
\]
which concludes the proof.   
\end{proof}

\begin{proof}[Proof of Lemma \ref{lemma:inter}]
    We provide a ``visual'' proof of this fact. Consider a ``histogram'' view of the voting procedure. Every point on the real line is assigned a score between 0 and 1, which is the fraction of sets that contained it. The resulting score curve is almost a kernel density estimate, except that it is not normalized: the ``density'' is given by $f(s) = \frac1{K}\sum_{k=1}^K \ind\{s \in \setc_k\}$. If this density is unimodal, then every level set will be an interval (and the level sets just correspond to $\setc^\tau$ for various $\tau$).  Now we argue that the input sets being intervals and $\cap_{k=1}^K \setc_k \neq \varnothing$ is sufficient for this unimodality. Take $\tau=1/2$ in what follows for simplicity, to focus on $\setc^M$. First note that if $\cap_{k=1}^K \setc_k \neq \varnothing$, then $\cap_{k=1}^K \setc_k$ is itself an interval (being the intersection of convex sets, it must be convex), and it must be contained in $\setc^M$. Starting from this interval, when we attempt to grow the interval by checking points just to the right (or left) of the current set, one will observe that the score of points can only decrease as we move further out. This is because our starting interval $\cap_{k=1}^K \setc_k$ lies ``inside'' every single input interval, and so as we move outwards from it, we can only hit closing endpoints of these intervals, slowly starting to exclude them one by one, but we can never hit an opening endpoint of an interval at a later point. This concludes the proof of unimodality, and hence of the lemma.
\end{proof}

\begin{proof}[Proof of Theorem \ref{thm:median-of-midpoints}]
    Assume for simplicity that all intervals (and hence midpoints) are distinct.
    Sort the intervals by their midpoints: let $\setc_{(k)}$ denote the $k$-th ordered set if its midpoint is $c_{(k)}$. We first note that if the intervals have the same width, then the intervals that contain the target $c$ must form a contiguous set according to this ordering. To elaborate, it is apparent that the covering intervals, if there are any, must be $\setc_{(a)},\setc_{(a+1)},\dots,\setc_{(b)}$ for some $1 \leq a \leq b \leq K$; indeed, it is not possible for $\setc_{(a)}$ and $\setc_{(a+2)}$ to cover without $\setc_{(a+1)}$ also covering (because they have the same width). 
    The above observation implies the following: if $> K/2$ intervals cover the target $c$, then so does the median interval. In particular, if $K$ is even then $c$ must be contained both in the intervals with midpoints $c_{(K/2)}$ and $c_{(1+ K/2)}$; in fact, if $c$ is exclusively contained within one of the two intervals, then it is contained at most in $K/2$ of the intervals (and not $K/2+1$).
    If $E$ denotes the event that $> K/2$ intervals contain the target $c$, we argued earlier that $\pr(E) \geq 1-2\alpha$. 
    More generally, if it happens to be the case that $>K/2$ intervals include an arbitrary point $s$, then the ``median interval'' $\setc^{(K/2)}$ must also contain $s$. Thus, $\setc^{(K/2)} \supseteq \setc^M$ as claimed. 

    To prove the final claim, let us assume for simplicity that all midpoints are distinct and all intervals are closed. Let $c_{(1)}, \dots, c_{(K)}$ be the ordered midpoints of $\setc_1, \dots, \setc_K$ and $\delta$ denote the width of the intervals. The boundaries of the intervals can be defined as
    \[
    a_{(k)} := c_{(k)} - \delta/2, \quad b_{(k)} := c_{(k)} + \delta/2, \quad k= 1, \dots, K.
    \]
    By definition, we have $a_{(1)} < \dots < a_{(K)}$ and $b_{(1)} < \dots < b_{(K)}$. In addition, since $\cap_{k=1}^K \setc_k \neq \varnothing$ then $\cap_{k=1}^K \setc_k$ must coincide with $[a_{(K)}, b_{(1)}]$. This implies that
    \[
    a_{(1)} < \dots < a_{(K)} < b_{(1)} < \dots < b_{(K)},
    \]
    and if $s \in [a_{(K)}, b_{(1)}]$ then it is \emph{voted} by all the intervals, in symbols, $\sum_{k=1}^K \ind\{s \in \setc_k \} = K$. With similar arguments, we find that if $s \in [a_{(K-1)}, a_{(K)})$ or $s \in (b_{(1)}, b_{(2)}]$ then $\sum_{k=1} \ind\{s \in \setc_k\}=K-1$; since $s$ will belong to all intervals except $\setc_K$ in the first case, while $s$ will belong to all intervals except $\setc_1$ in the second case. In general, we see that if $s \in [a_{(K-j)}, a_{(K-j+1)})$ or $s \in (b_{(j)}, b_{(j+1)}]$ then $\sum_{k=1}^K \ind\{s \in \setc_k\}=K-j$, for $j = 1, \dots, K-1$. This implies that the majority set $\setc^M$ is defined as
    \[
    \setc^M = 
    \begin{cases}
        \left(\cup_{j=\lceil K/2 \rceil}^{K-1} [a_{(j)}, a_{(j+1)}) \right) \bigcup [a_{(K)}, b_{(1)}] \bigcup \left(\cup_{j=1}^{\lceil K/2 \rceil - 1} (b_{(j)}, b_{(j+1)}]\right) = [a_{\lceil K/2 \rceil}, b_{\lceil K/2 \rceil}],&\text{if $K$ is odd}\\
        \left(\cup_{j= K/2 + 1 }^{K-1} [a_{(j)}, a_{(j+1)})\right) \bigcup [a_{(K)}, b_{(1)}] \bigcup \left(\cup_{j=1}^{ K/2 - 1} (b_{(j)}, b_{(j+1)}]\right) = [a_{K/2 +1 }, b_{K/2}], &\text{if $K$ is even}
    \end{cases}
    \]
    which concludes the claim.
\end{proof}

\begin{proof}[Proof of Theorem \ref{thm:length_mv}]
    The first part involves the observation that $\ind\{y > 1\} \leq y$ for $y\geq0$:
    \begin{equation*}
        \nu(\setc^\tau) = \int_\set \ind\left\{\frac{1}{K} \sum_{k=1}^K \ind\{x \in \setc_k\} > \tau \right\} d\nu(x) \leq \int_\set \frac{1}{K\tau} \sum_{k=1}^K \ind\{x \in \setc_k\} d\nu(x) = \frac{1}{K\tau} \sum_{k=1}^K \nu(\setc_k),
    \end{equation*}
    where intregration is with respect to the measure $\nu$.

    For the second part, let the endpoints of $\setc_1, \dots, \setc_K$ be denoted by $(a_1,b_1), \dots, (a_K,b_K)$. 
    Let $\Bar{m}_j:=(a_j + b_j)/2$ be the midpoint of the interval $\setc_j$, $j=1,\dots,K$. Without loss of generality, let us suppose that the length of the largest interval is $\max_k (b_k-a_k)=(b_1-a_1)=:\delta$. Let us define a new collection of intervals $\setc_1, \setc_2^\delta, \dots, \setc_K^\delta$, where the interval $\setc_j^\delta$ has endpoints
    \[
    a_j^\delta = \Bar{m}_j - \frac{\delta}{2}, \quad b_j^\delta = \Bar{m}_j + \frac{\delta}{2},
    \]
    for $j=2, \dots, K$. This implies that $\setc_j \subseteq \setc_j^\delta$ and intervals $\setc_1, \setc_2^\delta, \dots, \setc_K^\delta$ have the same width. Now, we can prove that $\setc^M(\setc_1,\dots,\setc_K) \subseteq \setc^M(\setc_1, \setc_2^\delta, \dots, \setc_K^\delta)$. Indeed, if the point $s \in \setc^M(\setc_1,\dots,\setc_K)$ then it is contained in at least more than $K/2$ of the initial intervals, but by definition $s$ must be included in the same set of the new intervals since $\setc_j \subseteq \setc_j^\delta$. From Theorem~\ref{thm:median-of-midpoints}, we have $\setc^M(\setc_1, \setc_2^\delta, \dots, \setc_K^\delta) \subseteq \setc^{(K/2)}$. So, if $K$ is odd, then $\setc^{(K/2)}$ is an interval with width $\delta$, while if $K$ is even it corresponds to the intersection of two equal-sized intervals, so its width is $\leq \delta$. Taking advantage of the fact that $\setc^\tau$ is decreasing in $\tau$ (i.e., $\setc^{\tau_2} \subseteq \setc^{\tau_1}$ if $\tau_1 \leq \tau_2$) then $\nu(\setc^\tau) \leq \max_k \nu(\setc_k)$ continues to hold for all $\tau \in [\frac{1}{2}, 1)$.
\end{proof}

The next lemma will be needed to prove Proposition~\ref{prop:indep}. In the following, we denote $X \sim PB(p_1, \dots, p_K)$ as the binomial Poisson random variable distributed as the sum of $K$ independent Bernoulli random variables with parameters $p_1, \dots, p_K$.    

\begin{lemma}
\label{lemma:pb}
    Let $X \sim PB(p_1, \dots, p_K)$ and $Y\sim \binomial(K,p)$ then $X$ is stochastically larger than $Y$, $X \geq_{st} Y$, if 
    \begin{equation*}
        p^K \leq \prod_{k=1}^K p_k.
    \end{equation*}
\end{lemma}
The proof is given in \citet{booland2002} and discussed in \citet{tang2023}.

\begin{proof}[Proof of Proposition~\ref{prop:indep}]
Let $\setc_1, \dots, \setc_K$ be a collection of independent confidence sets for the parameter $c$. Then $\phi_k = \ind\{c \in \setc_k\}$ is a Bernoulli random variable with parameter $p_k \geq 1-\alpha$. In addition, $\phi_1, \dots, \phi_K$ are independent (transformation of independent quantities). Suppose that $p_k = 1-\alpha$, for all $k=1,\dots,K$, then
\begin{equation*}
    \pr(c \in \setc^M) = \pr\left(\sum_{k=1}^K \ind\{c \in \setc_k \} > Q_K(\alpha) \right) = \pr(S_K > Q_K(\alpha)) \geq 1-\alpha,
\end{equation*}
where $S_K = \sum_{k=1}^K \phi_k \sim \binomial(K, 1-\alpha)$ and $Q_K(\alpha)$ defined in~\eqref{eq:quantile}. If $p_k \geq 1-\alpha$, $k=1,\dots, K$, then $S_K$ is distributed as a Poisson binomial with parameters $p_1, \dots, p_K$ and $\prod_{k=1}^K p_k \geq (1-\alpha)^K$. This implies that $S_K$ is stochastically larger than a $\binomial(K, 1-\alpha)$ due to Lemma~\ref{lemma:pb}.
\end{proof}

\begin{proof}[Proof of Theorem \ref{thm:exch}]
    Let $\phi_k=\phi_k(Z,c)=\ind\{c \notin \setc_k\}$ be a Bernoulli random variable such that $\Ev[\phi_k] \leq \alpha$, $k=1,\dots,K$. Since the sequence $(\phi_1, \dots, \phi_K)$ is exchangeable, we have 
    \begin{equation*}
    \begin{split}
    \pr(c \notin \setc^E) &= \pr\left(\exists k \leq K: c \notin \setc^M(1:k)\right) = \pr\left(\exists k \leq K: \frac{1}{k}\sum_{j=1}^k \phi_j \geq \frac{1}{2} \right)  \leq 2 \Ev\left[ \phi_1 \right]  \leq 2\alpha,
    \end{split}
    \end{equation*}
    where the first inequality holds due to the exchangeable Markov inequality (EMI) by \cite{ramdas2023}. It is straightforward to see that $\setc^E \subseteq \setc^M$, since $\setc^M(1:K)$ coincides with $\setc^M$.
\end{proof}

\begin{proof}[Proof of Theorem \ref{th:CR}]
    Let $\phi_k=\ind\{c \notin \setc_k\}$ be a Bernoulli random variable such that $\Ev[\phi_k] \leq \alpha$, $k=1,\dots,K$. Then using Additive-randomized Markov inequality (AMI),
    \[
    \begin{split}
    \pr(c \notin \setc^W) &= \pr\left(\sum_{k=1}^K w_k \left(1 - \phi_k\right) \leq \frac{1}{2} + U/2 \right) = \pr\left(\sum_{k=1}^K w_k \phi_k \geq \frac{1}{2} - U/2 \right) \\ &\leq 2 \Ev\left[ \sum_{k=1}^K w_k \phi_k \right] = 2 \sum_{k=1}^K w_k \Ev\left[\phi_k \right] \leq 2 \alpha \sum_{k=1}^K w_k = 2\alpha,
    \end{split}
    \]
    which proves \eqref{eq: coverage_cr}.

    In order to prove \eqref{eq:length_cr}, we follow the same lines as in Theorem~\ref{thm:length_mv}. 
    In particular,
    \[
    \begin{split}
        \nu(\setc^W) &= \int_\set \ind\left\{\sum_{k=1}^K w_k \ind\{x \in \setc_k\} > \frac{1}{2} + \frac{u}{2}\right\}d\nu(x) \leq \int_\set \ind\left\{\sum_{k=1}^K w_k \ind\{x \in \setc_k\} > \frac{1}{2}\right\} d\nu(x) \\
        &\leq \int_\set 2 \sum_{k=1}^K w_k \ind\{x \in \setc_k\} d\nu(x) = 2 \sum_{k=1}^K w_k \nu(\setc_k),
        \end{split}
    \]
    which concludes the proof.
\end{proof}

\begin{proof}[Proof of Theorem \ref{thm:exch2}]
By direct calculation
\[
\begin{split}
    \pr(\exists t \geq 1: c \notin \setc^E(1:t)) &= \pr\left(\cup_{T \geq 1} \{\exists t \leq T: c \notin \setc^E(1:t) \} \right)
    = \lim_{T \to \infty} \pr(c \notin \setc^E(1:T)) \leq 2\alpha,
\end{split}
\]
where the last equality is due to the fact that $\setc^E(1:t)$ shrinks when $t$ increases, while the last inequality is due to Theorem \ref{thm:exch}.
\end{proof}

\begin{proof}[Proof of Theorem \ref{thm:point}]
We provide two proofs: direct and indirect, starting with the latter. For each $k = 1,\dots,K$, note that $\hat \theta_k \pm w(n,\alpha)$  is a $1-\alpha$ confidence interval for $\theta$. (If $w$ involves other unknown nuisance parameters, we can pretend that these are being constructed by an oracle). Now, Theorem~\ref{thm:median-of-midpoints} directly implies the result, by noting that the median of the midpoints of these $K$ intervals is exactly $\hat \theta_{(\lceil K/2 \rceil)}$. For a direct proof, note that for $\hat \theta_{(\lceil K/2 \rceil)}$ to be more than $w(n,\alpha)$ away from $\theta$, it would have to be the case that at least $\lceil K/2\rceil$ of the individual $\hat \theta_k$ estimators would have to be more than $w(n,\alpha)$ away from $\theta$, but the latter event happens with probability at most $2\alpha$, following the proof of the majority vote procedure. The exchangeability claim follows by an argument identical to Theorem~\ref{thm:exch}. 
\end{proof}

\begin{proof}[Proof of Proposition \ref{prop:inf-unc-set}]
    Let $\phi_\lambda = \ind\{c \notin \setc_\lambda\}$ be a Bernoulli random variable such that $\Ev[\phi_\lambda] \leq \alpha$, for each $\lambda \in \Lambda$. Then
    \[
    \begin{split}
    \pr(c\notin \setc^W) &= \pr\left(\int_\Lambda w(\lambda)(1 - \phi_\lambda) d\lambda \leq \frac{1}{2} + U/2 \right) = \pr\left(\int_\Lambda w(\lambda)\phi_\lambda d\lambda \geq \frac{1}{2} - U/2 \right) \\
    &\stackrel{(i)}{\leq} 2 \Ev\left[\int_\Lambda w(\lambda) \phi_\lambda d\lambda \right] \stackrel{(ii)}{=} 2 \int_\Lambda w(\lambda) \Ev[\phi_\lambda] d\lambda \leq 2\alpha,
    \end{split}
    \]
    where $(i)$ is due to the uniformly-randomized Markov inequality, while $(ii)$ is due to Fubini's theorem. With similar arguments, we have that
    \[
    \begin{split}
    \nu(\setc^W) &= \int_\set \ind\left\{ \int_\Lambda w(\lambda) \ind\{x \in \setc_\lambda\} d\lambda > \frac{1}{2} + \frac{u}{2} \right\} d\nu(x) \leq \int_\set \ind\left\{ \int_\Lambda w(\lambda) \ind\{x \in \setc_\lambda\} d\lambda\ > \frac{1}{2}\right\} d \nu(x)\\
    &\leq 2 \int_\set \int_\Lambda w(\lambda) \ind\{x \in \setc_\lambda\} d\lambda\, d\nu(x) = 2 \int_\Lambda w(\lambda) \int_\set \ind\{x \in \setc_\lambda\} d\nu(x) \,d\lambda = 2 \int_\Lambda w(\lambda) \nu(\setc_\lambda) d\lambda,
    \end{split}
    \]
    which concludes the proof.
\end{proof}

\begin{proof}[Proof of Proposition \ref{prop:crc}]
    Using Uniformly-randomized Markov inequality (UMI) it is possible to obtain,
    \begin{equation*}
    \begin{split}
        \pr(Y_{n+1} \notin \setc^W(X_{n+1})) &= \pr\left(\sum_{k=1}^K w_k \loss(C_k(X_{n+1}), Y_{n+1}) \geq U\frac{B}{2}\right)\\
        & \leq \frac{2}{B} \Ev\left[\sum_{k=1}^K w_k \loss(C_k(X_{n+1}), Y_{n+1}) \right] \leq \frac{2}{B}\alpha.
    \end{split}
    \end{equation*}
    The same holds true using Markov's inequality and choosing $w_k=\frac{1}{K}, k=1,\dots,K$. The risk can be bounded as follows,
    \begin{equation*}
    \begin{split}
        \Ev\bigg[\loss(\setc^W(X_{n+1}), Y_{n+1})\bigg] &=\int \loss(\setc^W(X_{n+1}), Y_{n+1}) dP_{XY}^{n+1}\\
        &= \int \loss(\setc^W(X_{n+1}), Y_{n+1}) \ind\{Y_{n+1} \notin \setc^W(X_{n+1})\} dP_{XY}^{n+1}\\
        &\leq B \int \ind\{Y_{n+1} \notin \setc^W(X_{n+1})\} dP_{XY}^{n+1}\\
        &= B\, \pr(Y_{n+1} \notin \setc^W(X_{n+1})) \leq 2\alpha.
    \end{split}
    \end{equation*}
    The same result can be obtained using $\setc^M(x)$.
\end{proof}

\section{Merging confidence distributions}\label{sec:merg_conf_dist}
As outlined in Section~\ref{sec:intro}, if the aggregator knows the \emph{confidence distribution} for each agent, then it could be straightforward to combine them in a single confidence distribution. In particular, the \emph{confidence distribution} can be conceptualized as the distribution derived from the p-values corresponding to each point in the parameter space. In particular, for each agent and each point $s$ in the parameter space, we have the corresponding p-value for the hypothesis $H_0: c=s$. The confidence distribution can be obtained when $c$ is the functional or parameter of interest of a given distribution, but also in the conformal prediction setting where it is common to work with so called conformal (or rank-based) p-values \citep{vovk2005, lei2018}. This suggests that, in order to derive the distribution of the aggregator, we can combine the p-values obtained by the $K$ different agents for each point $s$ using a valid p-merging function \citep{vovk2020}. In particular, a p-merging function is simply a function that takes as input $K$ different p-values and outputs a p-value.

The simple majority vote rule can be viewed as an inversion of the fact that for $K$ dependent p-values, $2\cdot \text{median}(p_1,\dots,p_K)$ yields a valid p-value~\citep{ruger1978}. 
To see this, let $p_k=p_k(z;s)$ be the observed p-value by the $k$-th agent for the hypothesis null $H_0:c=s$; then, using the duality between tests and confidence sets, we have that $\setc_k=\{s \in \set: p_{k}(z;s) > \alpha\}$. Suppose (for the sake of contradiction) that $ p_{(\lceil K/2 \rceil)}(z;s) \le \alpha$ and $s \in \setc^M$. This implies that
\[   
\begin{split}
    \frac{1}{K}& \sum_{k=1}^K \ind\{s \in \setc_k\} > \frac{1}{2} \implies \sum_{k=1}^K \ind\{p_k(z;s) > \alpha\} > \left\lfloor\frac{K}{2}\right\rfloor \implies p_{(\lceil K/2 \rceil)}(z;s) > \alpha,
\end{split}
\]
which contradicts the supposition, establishing the claim. More generally, \cite{ruger1978} showed that $(K/k)p_{(k)}$ is a valid p-value, where $p_{(k)}$ is the $k$-th order statistic, recovering the Bonferroni correction at $k=1$, the union at $k=K$, and the median rule for $k=K/2$.
\cite{gasparin2024} extend the result by proving that $(K/k)p_{(\lceil Uk\rceil)}$ is a valid combination rule, where $U$ is a uniform random variable independent of the data. In particular, this rule is always smaller or equal than the R\"uger combination since $U$ is almost surely smaller or equal than 1.

In Figure~\ref{fig:confidence_distr} we report an example of the \emph{confidence distribution} obtained using two times the median of p-values as a merging function and its randomized extension. In particular, the example refers to an iteration of the simulation scenario described in Section \ref{sec:privateci}.

\label{sec:app_conf_distr}
\begin{figure}
    \centering
    \includegraphics[scale=0.75]{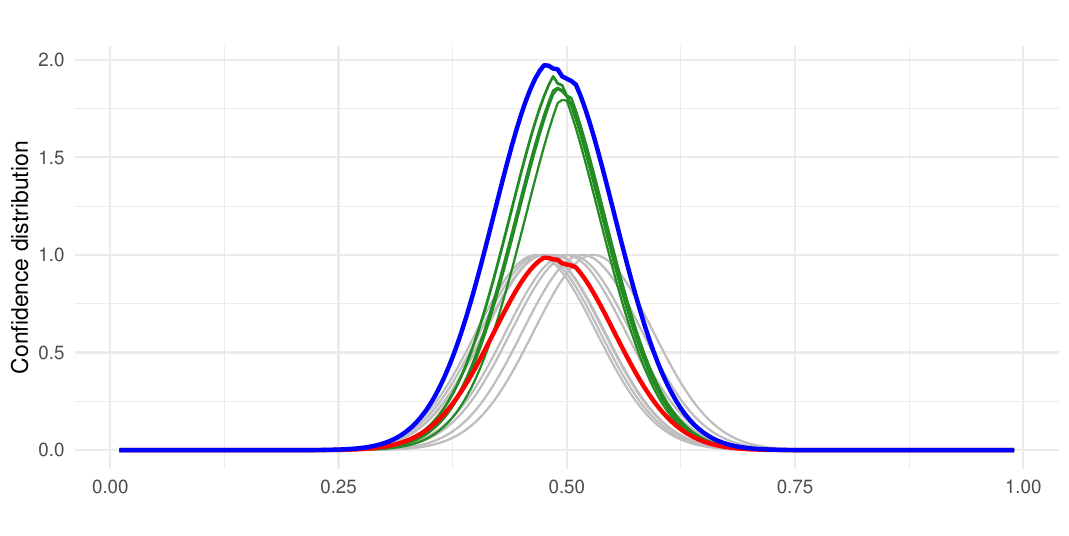}
    \caption{Example of \emph{confidence distributions} obtained in a simulation scenario (private multi-agent setting). Gray: confidence distributions of the single agents,  Red: confidence distribution of the median, Blue: confidence distribution of the median multiplied by 2, Green: four possible distributions obtained using the randomized version of R\"uger's combination;  we fix $u=\{0.2, 0.4, 0.6, 0.8\}$ in the latter method, though it would be random in practice. The red curve is not a valid combination of the grey ones, but the blue and green curves are valid.}
    \label{fig:confidence_distr}
\end{figure}

\section{Algorithm for interval construction}
\label{sec:algorithm}
As previously explained, the majority vote method can, in some cases, produce disjoint intervals; this stems from the fact that some of the $K$ intervals may have no common points. To address this, we propose an algorithm that returns the resulting set from the majority voting procedure. The starting point, for simplicity, is a collection of closed intervals, but it can be easily adapted to cases where intervals are open, or only some of them are open. In particular, the algorithm returns two vectors: one containing the lower bounds and the other containing the upper bounds. 

A naive solution involves dividing the space of interest $\set$ into a grid of points and evaluating how many intervals each point belongs to. However, this approach can become computationally burdensome, especially when the number of points is significantly high. Therefore, an alternative algorithm is recommended, which is based solely on the endpoints of the various intervals.

\begin{algorithm}[H]
\caption{Majority Vote Algorithm}
\label{alg:lower_bound_search}
\KwData{$K$ different intervals with lower bounds $a_k$ and upper bounds $b_k$, $k=1,\dots,K$; $w=(w_1,\dots,w_K)$ vector of weights; $\tau \in (0,1)$ threshold ($\tau=0.5$ for the majority vote procedure)}
\KwResult{lower; upper}
$q \leftarrow (q_1, \dots, q_{2K})$ vector containing the endpoints of the intervals in ascending order\;
$i\leftarrow1$\;
$\textrm{lower} \leftarrow \varnothing$\;
$\textrm{upper} \leftarrow \varnothing$\;
\While{$i < 2K$}{
    \If{$\sum_{k=1}^K w_k \ind\{a_k \leq \frac{q_{i} + q_{i+1}}{2} \leq b_k\} > \tau$}{
        $\textrm{lower} \leftarrow \textrm{lower} \cup q_i$\;
        $j \leftarrow i$\;
        \While{$\left(j < 2K\right) \textrm{and} \left(\sum_{k=1}^K w_k \ind\{a_k \leq \frac{q_{j} + q_{j+1}}{2} \leq b_k\} > \tau\right)$}{
        $j \leftarrow j+1$\;
        }
        $i \leftarrow j$\;
        $\textrm{upper} \leftarrow \textrm{upper} \cup q_{i}$;
    }
    \Else{
        $i \leftarrow i+1$\;
    }
}
\end{algorithm}

\section{Combining an infinite number of sets}\label{sec:inf_number}
There are instances where the cardinality of the set $K$ can be uncountably infinite. Consider a sequence of sets identified by some parameters --- such as in lasso regression, where each value of the penalty parameter corresponds to a distinct set. However, since the parameter can assume values on the positive semi-axis, the resulting number of sets is uncountable.

Specifically, we assume the existence of a mapping from $\lambda \in \Lambda \subseteq \mathbb{R}^d$ to $2^\set$, signifying that for each fixed $\lambda$ value, there exists a corresponding $1-\alpha$ uncertainty set $\setc_\lambda$. 
In addition, let us define a nonnegative \emph{weight function}, denoted as $w(\cdot)$, such that:
\begin{gather}
    \int_\Lambda w(\lambda) d\lambda = 1,
\end{gather}
and $w(\cdot)$ can be interpreted as a prior distribution on $\lambda$. In this case, we can define
\begin{equation}
\label{eq:setc_w_cont}
\setc^W = \left\{s \in \set: \int_\Lambda w(\lambda) \ind\{s \in \setc_\lambda\} d\lambda > \frac{1}{2} + U/2 \right\}.
\end{equation}

\begin{proposition}\label{prop:inf-unc-set}
   Let $\setc_\lambda$ be a sequence of $1-\alpha$ uncertainty sets indexed by $\lambda \in \Lambda$. Then the set $\setc^W$ defined in \eqref{eq:setc_w_cont} is a level $1-2\alpha$ uncertainty set.
   Furthermore, the measure associated with the set $\setc^W$ satisfies 
   \[
   \nu(\setc^W) \leq 2\int_\Lambda w(\lambda)\nu(\setc_\lambda)d\lambda.
   \]
\end{proposition}

In this scenario, in order to make the computation of $\setc^W$ computationally feasible, it is necessary to have a finite number of possible $\lambda$ values. A potential solution is to sample $N$ independent instances of $\lambda$ from the \emph{prior distribution}, calculate their respective $1-\alpha$ intervals, and then construct the set $\setc^R$ described in \eqref{eq:cr}. 

\section{Additional applications and simulations}
\label{sec:add_appl}
\subsection{Private multi-agent CI}
\label{sec:privateci}

As a case study, we employ the majority voting method in a situation where certain public data may be available to all agents, and certain private data may only be available to one (or a few) but not all agents. Consider a scenario involving $K$ distinct agents, each providing a \emph{locally differentially private} confidence interval for a common parameter of interest. As opposed to the centralized model of differential privacy in which the aggregator is trusted, local differential privacy is a stronger notion that does not assume a trusted aggregator, and privacy is guaranteed at an individual level at the source of the data. Further details about the definition of local privacy are not important for understanding this example; interested readers may consult \citet{dwork2014}. 

For $k=1, \dots, K$, suppose that the $k$-th agent has data about $n$ individuals $(X_{1,k}, \dots, X_{n,k})$ that they wish to keep locally private (we assume each agent has the same amount of data for simplicity). They construct their ``locally private interval'' based on the data $(Z_{1,k}, \dots, Z_{n,k})$, which represents privatized views of the original data. Suppose that an unknown fraction of the observations may be shared among agents, indicating that the reported confidence sets are not independent. An example of such a scenario could be a medical study, where each patient represents an observation, and a significant but unknown number of patients may be shared among different research institutions, or some amount of public data may be employed. Consequently, the confidence intervals generated by various research centers (agents) are not independent.

In the following, we refer to the scenario described in \cite{waudbysmith2023}, where the data $(X_{1,k}, .., X_{n,k}) \sim P$, and $P$ is any $[0,1]$-valued distribution with mean $\theta^*$. Data $(Z_{1,k},..,Z_{n,k})$ are  $\varepsilon$-locally differentially private ($\varepsilon$-LDP) views of the original data obtained using the nonparametric randomized response mechanism. 
The mechanism requires one additional parameter, $G$, which we set to a value of 1 for simplicity (the mechanism stochastically rounds the data onto a grid of size $G+1$, which is two in our case: the boundary points of 0 and 1).

A possible solution to construct a (locally private) confidence interval for the mean parameter $\theta^*$ is to use the locally private Hoeffding inequality as proposed in \cite{waudbysmith2023}. In particular, let $\hat{\theta}_k$ be the adjusted sample mean for agent $k$:
\[
\hat{\theta}_k := \frac{\sum_{i=1}^n z_{i,k} - n(1-r)/2}{nr},
\]
where $r:=(\exp\{\varepsilon\}-1)/(\exp\{\varepsilon\}+1)$. Then the interval 
\begin{equation*}
    \setc_k = \left[\hat{\theta}_k - \sqrt{\frac{-\log(\alpha/2)}{2nr^2}}, \hat{\theta}_k + \sqrt{\frac{-\log(\alpha/2)}{2nr^2}}\right]
\end{equation*}
is a valid $(1-\alpha)$-confidence interval for the mean $\theta^*$. 
The width of the confidence interval depends solely on the number of observations $n$, coverage $\alpha$, and privacy $\varepsilon$.

Once the various agents have provided their confidence sets, a non-trivial challenge may arise in merging them to obtain a unique interval for the parameter of interest. One possible solution is to use the majority-vote procedure described in the previous sections. We conducted a simulation study within this framework. In the first scenario, at each iteration, $n \times (K/2)$ observations were generated from a standard uniform random variable, and each agent was randomly assigned $n$ observations.  In the second scenario, the first agent had $n$ observations generated from a uniform random variable. For all agents with $k > 1$, a percentage $p$ of their observations was shared with the preceding agent, while the remaining portion was generated from $\mathrm{Unif(0,1)}$. In both scenarios, the number of agents is equal to 10, the number of observations ($n$) is common among the agents and equal to $100$, while the privacy parameter $\varepsilon$ is set to 2 (which is an appropriate value for local privacy, indeed Apple uses a value of 4 to collect data from iPhones; see \citet{appleEps}). The number of replications for each scenario is $10\,000$.

\begin{table}[ht]
\caption{Empirical average length of intervals and corresponding average coverage (within brackets) for the two simulation scenarios. In the second scenario, the percentage of shared observations is denoted as $p$. The $\alpha$-level is set to 0.1 --- the first column shows that the employed confidence interval is conservative (but tighter ones are more tedious to describe). The majority vote set is smaller than the individual ones, but it overcovers (despite the theoretically guarantee being one that permits some undercoverage), which is an intriguing phenomenon. The randomized majority vote method produces the smallest sets than the others while maintaining good coverage. Randomized voting is not very different from the original intervals.}
\label{tab:mult_ag}
\centering
\begin{tabular}{cc|ccccc}
  \hline
  Scenario & $p$ & $\setc_k$ & $\setc^M$ & $\setc^R$ & $\setc^U$ & $\setc^\pi$ \\ 
  \hline
  (i) & - & 0.3214 & 0.3058 & 0.2282 & 0.3212 & 0.3210  \\ 
  & & (0.9880)& (1.0000) & (0.9752) & (0.9892) & (0.9892)  \\ 
  (ii) & 0.5 & 0.3214 & 0.3062 & 0.2302 & 0.3218 & 0.3215  \\ 
  & & (0.9877) & (1.0000) & (0.9749) & (0.9879) & (0.9879)  \\ 
  (ii) & 0.75 & 0.3214 & 0.3062 & 0.2302 & 0.3218 & 0.3215  \\ 
  & & (0.9877) & (1.0000) & (0.9749) & (0.9879) & (0.9879)  \\ 
  (ii) & 0.9 & 0.3214 & 0.3063 & 0.2297 & 0.3223 & 0.2319 \\ 
  & & (0.9877) & (1.0000) & (0.9764) & (0.9870) & (0.9775) \\ 
   \hline
\end{tabular}
\end{table}

As can be seen in Table~\ref{tab:mult_ag}, the length of the intervals constructed by the agents remains constant throughout the simulations, since the values of $n$ and $\varepsilon$ remain unchanged. In contrast, the intervals formed by the majority and randomized majority methods are smaller, compared to those constructed by individual agents. The coverage level achieved by individual agents' intervals (first column) significantly exceeds the threshold of $1-\alpha$, but this is expected since the intervals are nonasymptotically valid and conservative. The coverage derived from the majority method is notably high, approaching 1. The incorporation of a randomization greatly reduces the length of the sets while maintaining coverage at a slightly lower level than that of single agent-based intervals. The use of the randomized union introduced in~\eqref{eq:cu} produces sets with nearly the same lenght and coverage as the ones produced by single agents. 
If the aggregator had access to the $(1-\alpha)$-confidence intervals level for all possible $\alpha$, it would be able to derive the confidence distribution. 
The actual values of the length and coverage should not be given too much attention: there are other, more sophisticated, intervals derived in the aforementioned paper (empirical-Bernstein, or asymptotic) and these would have shorter lengths and less conservative coverage, but they take more effort to describe here in self-contained manner and were thus omitted.

\subsection{Derandomizing cross-fitting in causal inference}\label{sec:cross-fitted}
In many problems of semiparametric inference, it is common to be interested in making inference about a parameter $\theta$ in the presence of nuisance parameters $\xi$, possibly of high dimension. Often, the estimation of $\xi$ is carried out using machine learning methods that tend to perform adequately well in high-dimensional settings. However, issues such as overfitting and regularization bias can pose challenges for the inference of the parameter of interest. A solution is proposed in \cite{chernozhukov2018double}, where cross-fitting is employed to circumvent such problems. Subsequently, we will discuss this technique in the context of estimating the Average Treatment Effect (ATE).

We now define the problem setup and the notation, suppose that we observe a sample of $n$ iid random variables $Z_1, \dots, Z_n$ from a distribution $P$. In particular, $Z_i$ consists on the triplet $Z_i:=(X_i, A_i, Y_i)$; where $X_i \in \mathbb{R}^d$ are the baseline covariates of the $i$-th observation, $A_i$ is the treatment that they receive while $Y_i \in \mathbb{R}$ is the outcome observed after the treatment. As an example, consider observations that represent a cohort of patients undergoing two different types of treatment (each patient receives only one of the two treatments), where their covariates include control variables such as age and sex. Our target parameter is represented by the ATE,
\(
\theta := \Ev[Y^1 - Y^0],
\)
where $Y^a$ represents the counterfactual outcome for a randomly selected subject had they received the treatment; see \citet[Chapter 2]{vanderweele2015explanation} and \cite{robins1992identifiability} for an introduction. Subject to conventional causal identification assumptions, commonly known as consistency, positivity, and no unmeasured confounding (refer, for instance, to \cite{kennedy2016}), it follows that $\theta$ can be expressed as a functional of the distribution $P$, without invoking counterfactual considerations. In particular,
we have
\[
\theta = \theta(P) = \Ev\left\{\Ev[Y \mid A=1, X=x] - \Ev[Y \mid A=0, X=x] \right\}.
\]
Once the data have been observed, we have to find an efficient estimator of ATE. First of all, let us define $h(z)$ as
\[
h(z) := \left(\mu^1(x) - \mu^0(x)\right) + \left(\frac{a}{\pi(x)} - \frac{1-a}{1-\pi(x)}\right) \left(y - \mu^a (x)\right),
\]
where $\mu^a(x):= \Ev[Y\mid X=x, A=a]$ is the regression function for observations whose treatment level is equal to $a=\{0,1\}$, while $\pi(x):=\pr(A=1 \mid X=x)$ is the so-called propensity score. The \emph{nuisance} functions $\xi=(\mu^1(x),\mu^0(x), \pi(x))$ are unknown, and need to be estimated even if not of direct interest (actually in RCTs $\pi(x)$ is known and only $\mu^a(x), a=\{0,1\},$ need to be estimated). An efficient estimator for the ATE is given by $n^{-1} \sum_{i=1}^n h(z_i)$; see \cite{kennedy2016}.

Adopting a parametric structure for nuisance functions can be restrictive, especially in high-dimensional contexts. Frequently, these functions are estimated through machine learning methods, which, however, encounter challenges such as overfitting and selection bias and render the convergence complex. As described in \cite{chernozhukov2018double}, a potential solution to mitigate these issues is the use of the cross-fit estimator, using a random data splitting approach. In particular, the functions $\xi$ are estimated on a first (random) part of the dataset, called $\mathcal{Z}^{trn}$, in order to obtain $\hat{\xi}$ and an estimator of $\theta$ is given by $\hat{\theta}_1 = |\mathcal{Z}^{eval}|^{-1}\sum_{i\in\mathcal{Z}^{eval}} h(z_i)$ where $\mathcal{Z}^{eval}$ is $\{Z_i\}_{i=1}^n \smallsetminus \mathcal{Z}^{trn}$ and $\xi$ is substituted by its estimated counterpart. Another estimator, $\hat{\theta}_2$, is simply obtained by switching the roles of the set $\mathcal{Z}^{trn}$ and $\mathcal{Z}^{eval}$. The cross fit estimator is simply
\(
\hat{\theta} = \frac{\hat \theta_1 + \hat \theta_2}{2}.
\)
In general, if observations are partitioned into $B$ non-overlapping samples, then one can obtain $B$ different estimators and, in this case, the cross-fit estimator is given by their average. Under suitable conditions, the cross-fit estimator is asymptotically normal, so the computation of a confidence interval is possible. Also in this case, the estimator and the intervals depend on the random split of the data. One can repeat the procedure $K$ times in order to construct $K$ different intervals and subsequently merge them using the majority vote procedure.

\begin{figure}[h!]
    \centering
    \includegraphics[width=0.8\textwidth]{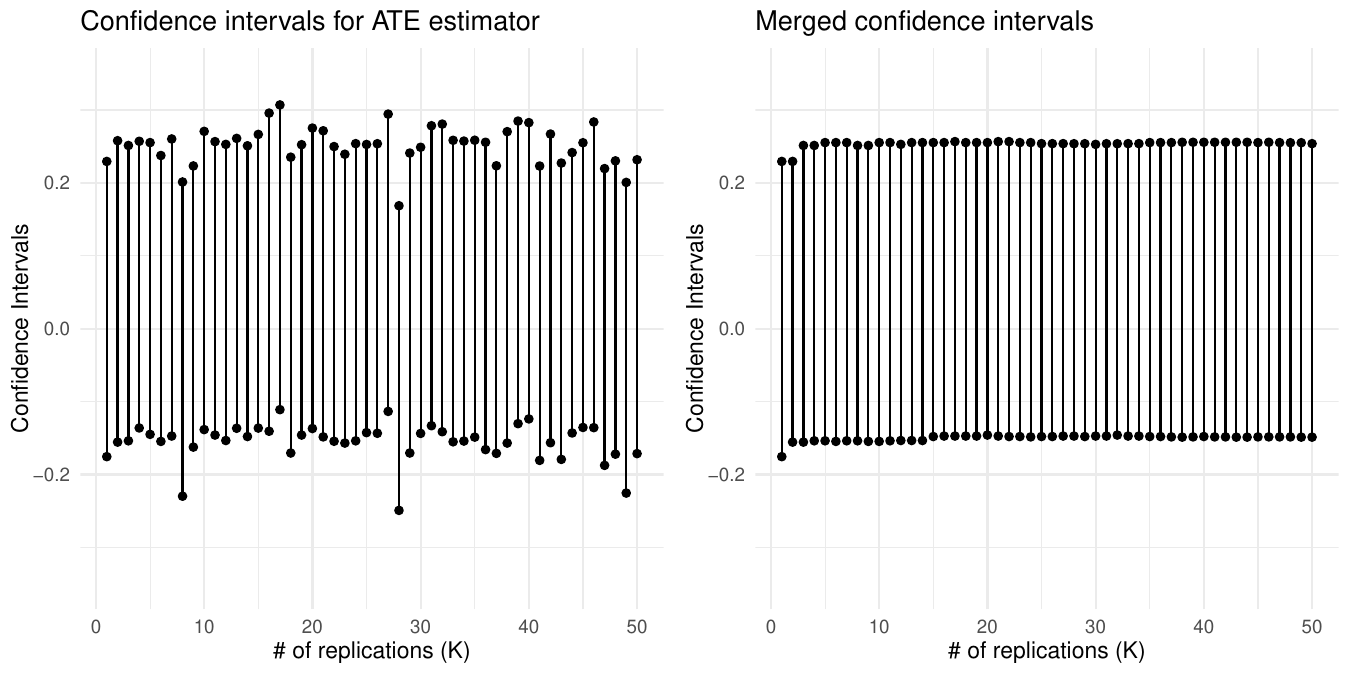}
    \caption{The left plot shows 50 estimates of the ATE via cross-fitting using different sample splits on the same dataset. The right plot shows their combination via majority vote. Both plots can be produced one-by-one from left to right, and stopped when the intersection of the sets in the right plot is deemed stable enough.}
    \label{fig:causal_plot}
\end{figure}

\paragraph{Real data application.} We investigate the impact of $K$ in an example using real data on the effectiveness of a treatment for type 2 diabetes. Data are used in \cite{berchialla2022prediction} and are publicly available\footnote{https://datadryad.org/resource/doi:10.5061/dryad.qt743/2}. Specifically, the study contains $n=92$ patients, and for each patient a series of baseline covariates is measured. Treatment refers to the drug administered to the patient. Specifically, \texttt{glimepiride} serves as the baseline drug, while \texttt{sitagliptin} represents the novel drug under evaluation. It is noteworthy that the study does not adhere strictly to a RCT design, as the allocation of medications to participants is not random, but based on specific characteristics of the participants. The outcome is represented by the difference of the level of \texttt{HbA1c} before and after treatment, and takes a value of one if the difference is less than -0.5, and zero otherwise. Due to the difficulty in specifying a parametric model for $\mu^a(x)$ and $\pi(x)$, they were estimated using machine learning algorithms. In this case, a random forest is employed to estimate $\mu^a(x)$ while a penalized logistic regression is used to estimate the propensity score. The number of non-overlapping subsets $B$ is fixed at 4, as suggested in \cite{chernozhukov2018double}; and the cross-fit estimator is computed using the R package \texttt{DoubleML} \citep{doubleml2024R}. As it is possible to see from Figure \ref{fig:causal_plot} the intervals obtained for ATE can differ between them, while intervals obtained using the majority vote remain essentially the same for $K>15$. Given that the sets are obtained with the same generating mechanism but using different random data splits, they are exchangeable and the set $\setc^E$ is given by the running intersection of the majority vote sets.

\subsection{Additional simulations with different vector of weights}\label{sec:add_sims_w}
In this section we study the impact of different weight vectors in our proposed method through a simulation study. In particular, data are simulated as in Section~\ref{sec:conf_sim}, and five different vector of weights are used. The five different cases can be summarized as follows: (i) the weights are uniform and this case coincide with the standard majority vote; (ii) the weight assigned to the last ten sets is thrice the weight assigned to the first ten, i.e.\ $w \propto (1,\dots,1,3,\dots,3)$; (iii) the scenario is opposite to the earlier one and $w=(3,\dots,3,1,\dots,1)$; (iv) the weights are increasing and $w\propto (1,2,\dots,20)$; (v) the weights are decreasing and $w\propto (20,19,\dots,1)$. The process is repeated $5000$ times and in Table~\ref{tab:sim_weights} we report the empirical coverage and the empirical size of the set $\setc^W$ without randomization.

\begin{table}
\centering
\caption{Empirical coverage and empirical size of the sets $\setc^W$ using different weights. The size is smaller if the weights are more concentrated in the first positions (case (iii) and case (v)).}
\begin{tabular}{ccc}
  \hline
 Case & Coverage & Width \\ 
  \hline
  (i) & 0.985 & 8.719 \\ 
  (ii) & 0.981 & 11.832 \\ 
  (iii) & 0.987 & 7.681 \\ 
  (iv) & 0.979 & 12.112 \\ 
  (v) & 0.987 & 7.676 \\ 
   \hline
\end{tabular}
\label{tab:sim_weights}
\end{table}

From Table~\ref{tab:sim_weights} it can be observed that the empirical coverage of the methods appears to be consistent across different cases; in contrast, there are substantial differences in terms of interval width. Specifically, if the weights are more concentrated in the first part of the vector, narrower intervals are obtained. This result is in line with Figure~\ref{fig:lasso_sims}, as smaller penalty parameters seem to correspond to narrower intervals.

As expected, the weights plays a key role, like that of a prior, and a ``bad'' prior can enlarge  the size of the weighted majority vote sets. However, the coverage guarantees remain unaltered as pointed out in Section~\ref{sec:majority_vote}. 
If there is no prior knowledge regarding the different methods used to obtain the various intervals, it seems prudent to use a vector of uniform weights, but other (more complicated) options exist. In the considered example, one might set aside a portion of the dataset to select an appropriate penalty parameter and subsequently use this parameter to obtain the prediction interval on the portion not used for selection. In a similar spirit, rather than choosing a single parameter, one could opt for multiple parameters, assigning weights inversely proportional to the error metric observed on the selection set. This can be particularly convenient if there is no a clear winner. Outside the considered scenario, the prior weights can be elicited according to the information available to the user (like the sample sizes used by the different agents or the expected performance of the different algorithms).

\section{Merging sets with conformal risk control}\label{sec:conformal-risk-control}

In Section~\ref{sec:conformal}, we have used conformal prediction to obtain prediction intervals that allow the derivation of a lower bound for the probability of miscoverage. However, in many machine learning problems, miscoverage is not the primary and natural error metric, as explained in \cite{angelopoulos2021}. Consequently, a more general metric may be necessary to assess the loss between the target of interest and an arbitrary set $\setc$. To achieve this, one may proceed by choosing a loss function
\begin{equation}
\label{eq:loss_general}
    \loss: 2^{\mathcal{Y}} \times \mathcal{Y} \to [0, B], \quad B \in (0, \infty),
\end{equation}
where $\mathcal{Y}$ is the space of the target being predicted, and $2^{\mathcal{Y}}$ is the power set of $\mathcal{Y}$. In addition, we require that the loss function satisfies the following properties:
\begin{gather*}
    \setc \subset \setc' \implies \loss(\setc, c) \geq \loss(\setc', c),\\
    \loss(\setc, c) = 0 \,\, \mathrm{if }\,\, c \in \setc.
\end{gather*}
By definition, the loss function in~\eqref{eq:loss_general} is bounded and shrinks if $\setc$ grows (eventually shrinking to zero when the set contains the target). Similar to the conformal prediction framework described in Section~\ref{sec:conformal}, we consider the target of interest as $Y_{n+1}\in \mathcal{Y}$, while $\setc=\setc(x),x\in\mathcal{X},$ is a set based on an observed collection of feature-response instances $z_i = (x_i, y_i),\, i = 1, \dots, n$. 

\cite{angelopoulos2022} generalize (split) conformal prediction to prediction tasks where the natural notion of error is defined by a loss function that can be different from miscoverage. In particular, their extension of conformal prediction provides guarantees of the form
\begin{equation}
\label{eq:loss_guarantee}
\Ev\bigg[\loss\left(\setc(X_{n+1}), Y_{n+1}\right) \bigg] \leq \alpha,
\end{equation}
where $\alpha \in (0,B)$. It can be seen that \emph{standard} conformal prediction intervals can be obtained simply by choosing $\loss\left(\setc(X_{n+1}), Y_{n+1}\right) = \ind\{ Y_{n+1} \notin \setc(X_{n+1})\}$.

\subsection{Majority vote for conformal risk control}

It appears possible to extend the majority vote procedure, described in Section~\ref{sec:gen_maj_vote} and in Section~\ref{sec:randomized_maj_vote}, to sets with a conformal risk control guarantee.

\begin{proposition}\label{prop:crc}
Let $\setc_1(x), \dots, \setc_K(x)$ be $K \geq 2$ different sets with the property in~\eqref{eq:loss_guarantee}, $x \in \mathcal{X}$ and $w=(w_1, \dots, w_K)$ a vector of weights defined as in~\eqref{eq:weights}. Define
    \begin{gather}
        \setc^M(x) = \left\{y \in \mathcal{Y}: \frac{1}{K} \sum_{k=1}^K \loss\left(\setc_k(x), y\right) < \frac{B}{2}\right\},\\
        \label{eq:cr_loss}
        \setc^W(x) = \left\{y \in \mathcal{Y}: \sum_{k=1}^K w_k \loss\left(\setc_k(x), y\right) < U\frac{B}{2}\right\},
    \end{gather}
where $U \sim \mathrm{Unif}(0, 1)$. Then, $\setc^M(x)$ and $\setc^W(x)$ control the conformal risk at level $2\alpha$.
\end{proposition}
The proof, reported in Appendix \ref{sec:app_proof}, initially involves calculating the miscoverage of the set, followed by establishing an upper bound for the risk defined as the expected value of the loss function.

The obtained bound may be excessively conservative, as it involves substituting the value of the loss function with its upper limit. Consequently, the resulting sets can be too large, particularly when the loss function is uniform over the interval $[0, B]$ or centered on an internal point, or exhibits skewness towards smaller values. 

\subsection{Experiment on simulated data}
We now present a simulation study regarding majority vote for conformal risk control. In classification problems, it often occurs that misclassified labels may incur a different cost based on their importance. An example of a loss function used for this purpose is 
\begin{equation*}
    \loss(\setc, y) = L_y\ind\{y \notin \setc\},
\end{equation*}
where $L_y$ is the cost related to the misclasification of the label $y \in \mathcal{Y}$ and $\mathcal{Y}$, in this case, denotes the finite set of possible labels. 

The methodology introduced by \citet{angelopoulos2022} uses predictions generated from a model $\hat{\mu}$ to formulate a function $\setc_\gamma(\cdot)$ that assigns features $x \in \mathcal{X}$ to a set. The parameter $\gamma$ denotes the degree of conservatism in the function, with smaller values of $\gamma$ producing less conservative outputs. The primary objective of their approach is to infer the value of $\gamma$ using a calibration set, with the aim of achieving the guarantee outlined in \eqref{eq:loss_guarantee}. Given an error threshold $\alpha$, \cite{angelopoulos2022} define
\begin{equation*}
\hat{\gamma} := \inf\left\{\gamma: \frac{n}{n+1} \sum_{i=1}^n \loss(\setc_\gamma(x_i), y_i) + \frac{B}{n+1} \leq \alpha \right\}.
\end{equation*}
For classification problems, $\setc_\gamma(x_i)$ is simply expressed as $\setc_\gamma(x_i) = \{y \in \mathcal{Y}: \hat{\mu}(x_i)_y \geq 1-\gamma\}$, where $\hat{\mu}(x_i)_y$ represents the probability assigned to the label $y$ by the model.

\begin{figure}
    \centering
    \includegraphics[width=0.8\textwidth]{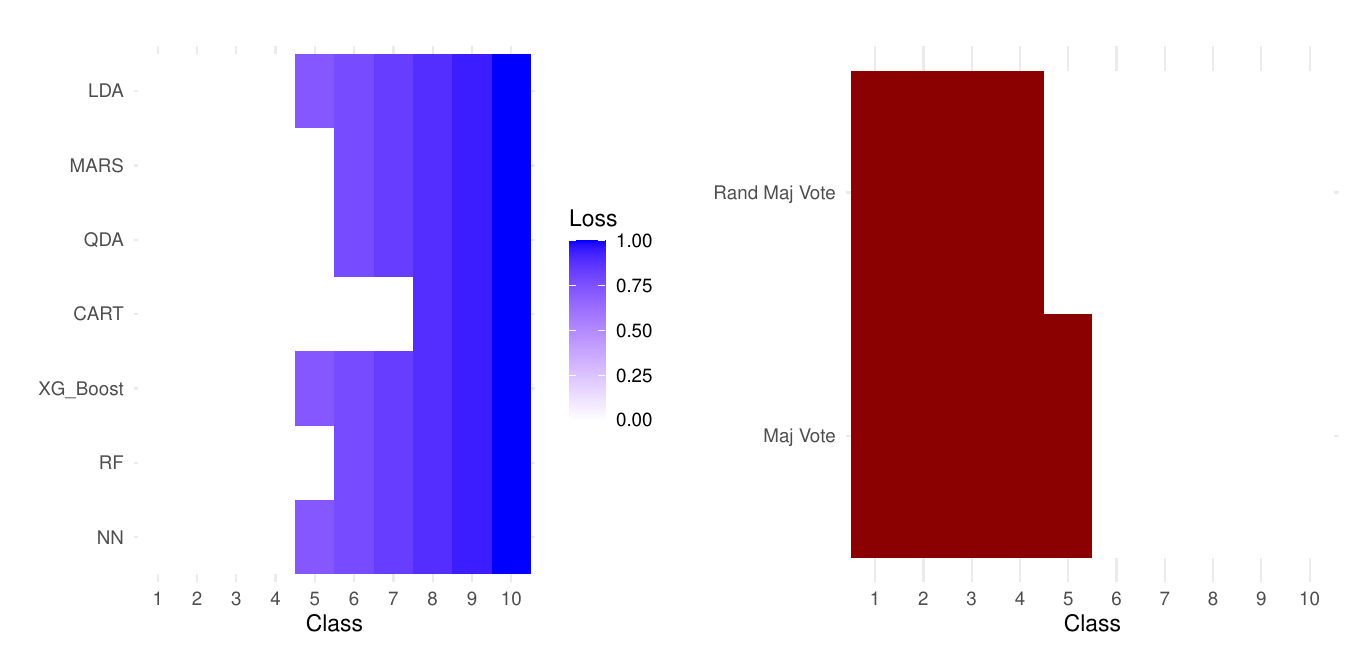}
    \caption{Left: losses of various algorithms, where the loss is zero if the point is included in the interval. Right: points included by majority vote and randomized majority vote.}
    \label{fig:rcps_set}
\end{figure}

The approach is used in a classification task with simulated data. We simulated data from 10 classes, each originating from a bivariate normal distribution with a mean vector $(i, i)$ and a randomly generated covariance matrix, where $i=1,\dots,10$. In addition, two covariates were incorporated to add noise. For each class, we generated 600 data points, partitioning them into three equal subsets: one-third for the estimation set, one-third for the calibration set, and one-third for the test set. Loss values are represented by $L_y = \frac{8+y}{18}$, where $y \in \{1,\dots,10\}$. This indicates that the cost of misclassifying a label in the last class is twice that of a label in the first class. We used $K=7$ different classification algorithms, and the parameters $\hat{\gamma}_k$ were estimated within the calibration set, for all $k=1, \dots, K$. An example of the majority vote procedure is shown in Figure~\ref{fig:rcps_set}. The empirical losses computed in the test set of the methods are, respectively, 0.042 for the simple majority vote and 0.084 for the randomized version of the method. It is important to highlight that, in some situations, the majority vote procedure can produce too large sets. Suppose that the loss for a single point is less than half; then the procedure will include the point also if it is not included in any of the sets. The randomized method can present the same problem if the values of the loss function are close to zero. A possible solution to mitigate the problem is to tune the threshold parameter of the majority vote to a smaller value achieving different levels of guarantee.

\end{document}